\newif\ifarxiv
\arxivtrue

\PassOptionsToPackage{prologue,dvipsnames}{xcolor}

\ifarxiv
\documentclass[acmsmall,nonacm]{acmart}
\else
\documentclass[acmsmall]{acmart}
\fi

\setcopyright{cc}
\setcctype{by}
\acmDOI{10.1145/3729281}
\acmYear{2025}
\acmJournal{PACMPL}
\acmVolume{9}
\acmNumber{PLDI}
\acmArticle{178}
\acmMonth{6}
\received{2024-11-14}
\received[accepted]{2025-03-06}

\usepackage{mathpartir}
\usepackage{tikz-cd}
\usepackage{lipsum}
\usepackage{enumitem}
\usepackage{wrapfig}
\usepackage{fancyvrb}
\usepackage{xspace}
\usepackage{thmtools} 
\usepackage{hyperref}
\usepackage[capitalise]{cleveref}

\usepackage[LGR,T1]{fontenc}
\DeclareMathAlphabet{\mathgtt}{LGR}{cmtt}{m}{n}

\renewcommand{\Gamma}{\mathgtt{G}}
\renewcommand{\Delta}{\mathgtt{D}}
\renewcommand{\Sigma}{\mathgtt{S}}
\renewcommand{\mu}{\mathgtt{m}}

\usepackage{mathastext}

\colorlet{keywordcolor}{blue}
\colorlet{nonlintycolor}{RubineRed}
\colorlet{lintycolor}{PineGreen}
\colorlet{nonlinconstructorcolor}{nonlintycolor}
\colorlet{linconstructorcolor}{lintycolor}
\colorlet{alphabetcolor}{black}
\colorlet{funarrowcolor}{lintycolor}
\definecolor{backcolour}{rgb}{0.95,0.95,0.92}

\newcommand{\uparrowcode}{\color{nonlinconstructorcolor} \uparrow}
\newcommand{\tocode}{\color{nonlinconstructorcolor} \to}
\newcommand{\tocodelin}{\color{linconstructorcolor} \to}

\usepackage{listings}
\lstdefinelanguage{Agda}{
alsoletter={0123\&'},
keywords={data, where, module, import, open, public,
record, field, let, in, if, then, else, case, of, with,
do, postulate, primitive, mutual, abstract, private,
forall, exists, cong, set, prop, sort, Level, Data, Type,
Renamer},
morekeywords=[2]{L, Bool, Nat, Fin},
morekeywords=[3]{String, Dyck, Trace, Exp, Atom, O, D, C, A, NotStartsWithLP, NotStartsWithRP},
morekeywords=[4]{bal, .nil, .cons, cons, nil, inl, inr, stop, 1to1, 1to2, 0to2, 0to1,
  done, add, num, parens, left, unexpected, lookAheadRP, lookAheadNot, closeGood, closeBad,
  doneGood, doneBad},
morekeywords=[5]{.states, states, .isAcc, isAcc, transitions, src, dst, .transitions, .dst,
  .src, .label, label, 0 , 1 , 2, 3, true, false},
morekeywords=[6]{NUM},
morekeywords=[7]{\&, '},
sensitive=true,
comment=[l]{--},
morecomment=[s]{\{-}{-\}},
morestring=[b]",
mathescape=true,
moredelim={[is][]{\#}{\#}},
escapeinside={(*@}{@*)},
escapebegin=\color{funarrowcolor},
escapeend={}
}
\usepackage{placeins} 

\lstnewenvironment{floatlisting}{
  \FloatBarrier
  \lstset{
    aboveskip=2mm
  }
}{}

\usepackage{stmaryrd}
\usepackage{scalerel}
\usepackage{tikz}
\usetikzlibrary{automata, positioning, arrows, fit}

\usepackage{pdfpages}

\newcommand{\Subst}{\textrm{Subst}}
\newcommand{\SPF}{SPF}
\newcommand{\Var}{\mathsf{Var}}
\newcommand{\K}{\mathsf{K}}
\newcommand{\map}{\mathsf{map}}
\newcommand{\roll}{\mathsf{roll}}
\newcommand{\fold}{\mathsf{fold}}
\newcommand{\inl}{\mathsf{inl}}
\newcommand{\inr}{\mathsf{inr}}
\newcommand{\sem}[1]{\left\llbracket{#1}\right\rrbracket}
\newcommand{\semg}[1]{\sem{#1}\gamma}

\newcommand{\lto}{\multimap}
\newcommand{\tol}{\mathrel{\rotatebox[origin=c]{180}{$\lto$}}}

\newcommand{\StringSem}{\mathbf{String}}
\newcommand{\CharGram}{\texttt{Char}}
\newcommand{\StringGram}{\texttt{String}}

\newcommand{\Set}{\mathbf{Set}}
\newcommand{\hSet}{\mathbf{hSet}}

\newcommand{\Grammar}{\mathbf{Gr}}
\newcommand{\semcat}{\mathbf{C}}

\newcommand{\Bool}{\mathtt{Bool}}
\newcommand{\true}{\mathtt{true}}

\newcommand{\theoryname}{Dependent Lambek Calculus\xspace}
\newcommand{\theoryabbv}{$\textrm{Lambek}^D$\xspace}
\newcommand{\lnld}{$\textrm{LNL}_D$}

\newcommand{\isTy}{\textrm{ type}}
\newcommand{\isCtx}{\textrm{ ctx}}
\newcommand{\isSmall}{\textrm{ small}}
\newcommand{\isLinTy}{\textrm{ lin. type}}
\newcommand{\isSmallLin}{\textrm{ small lin.}}
\newcommand{\isLinCtx}{\textrm{ lin. ctx.}}

\newcommand{\quoteTy}[1]{\lceil{#1}\rceil}
\newcommand{\unquoteTy}[1]{\lfloor{#1}\rfloor}

\DeclareMathOperator*{\bigamp}{\scalerel*{\&}{\bigoplus}}
\DeclareMathOperator*{\bigwith}{\scalerel*{\&}{\bigoplus}}

\newcommand{\ltonl}[1]{~\uparrow #1}
\newcommand{\nil}{\texttt{nil}}

\newcommand{\cons}{\texttt{cons}}
\newcommand{\epscons}{\epsilon\texttt{cons}}

\newcommand{\Trace}{\texttt{Trace}}
\newcommand{\literal}[1]{\texttt{\textquotesingle#1\textquotesingle}}
\newcommand{\stringquote}[1]{\texttt{\textquotedbl#1\textquotedbl}}
\newcommand{\internalize}[1]{\lceil#1\rceil}

\newcommand{\oplusinj}[2]{\sigma\,#1\,#2}
\newcommand{\withprj}[2]{\pi\,#1\,#2}

\newcommand{\simulsubst}[2]{#1\{#2\}}
\newcommand{\subst}[3]{\simulsubst {#1} {#2/#3}}
\newcommand{\el}{\mathsf{el}}
\newcommand{\letin}[3]{\mathsf{let}\, #1 = #2 \, \mathsf{in}\, #3}
\newcommand{\lamb}[2]{\lambda #1.\, #2}
\newcommand{\lamblto}[2]{\lambda^{{\lto}} #1.\, #2}
\newcommand{\lambtol}[2]{\lambda^{{\tol}} #1.\, #2}
\newcommand{\dlamb}[2]{{\lambda}^{{\&}} #1.\, #2}
\newcommand{\withlamb}[2]{{\lambda}^{{\&}} #1.\, #2}
\newcommand{\app}[2]{#1 \, #2}
\newcommand{\applto}[2]{#1 \, #2}
\newcommand{\apptol}[2]{#1 \mathop{{}^{\tol}} #2}
\newcommand{\PiTy}[3]{\textstyle\prod_{#1 : #2} #3}
\newcommand{\SigTy}[3]{\textstyle\sum_{#1 : #2} #3}
\newcommand{\PiTyLimit}[3]{\textstyle\prod\limits_{#1 : #2} #3}
\newcommand{\SigTyLimit}[3]{\textstyle\sum\limits_{#1 : #2} #3}
\newcommand{\LinPiTy}[3]{\textstyle\bigamp_{#1 : #2} #3}
\newcommand{\LinPiTyLimit}[3]{\bigwith\limits_{#1 : #2} #3}
\newcommand{\LinSigTy}[3]{\textstyle\bigoplus_{#1 : #2} #3}
\newcommand{\LinSigTyLimit}[3]{\bigoplus\limits_{#1 : #2} #3}

\newcommand{\equalizer}[3]{\{#1\,|\,\applto {#2}{#1} = \applto{#3}{#1} \}}
\newcommand{\equalizerin}[1]{\langle #1 \rangle}
\newcommand{\equalizerpi}[1]{#1.\pi}

\newcommand{\ctxwff}[1]{#1 \isCtx}
\newcommand{\ctxwffjdg}[2]{#1 \vdash #2 \isTy}
\newcommand{\linctxwff}[2]{#1 \vdash #2 \isLinCtx}
\newcommand{\linctxwffjdg}[2]{#1 \vdash #2 \isLinTy}
\newcommand{\nonlinterm}[3]{#1 \vdash #2 : #3}
\newcommand{\linterm}[4]{#1 ; #2 \vdash #3 : #4}
\newcommand{\nonlineq}[4]{#1 \vdash #2 \equiv #3 : #4}

\newcommand{\LLL}{\textrm{LL}}
\newcommand{\LRR}{\textrm{LR}}
\newcommand{\LALRR}{\textrm{LALR}}
\newcommand{\LL}[1]{\LLL(#1)}
\newcommand{\LR}[1]{\LRR(#1)}

\newsavebox{\logoagdabox}
\sbox{\logoagdabox}{%
  \raisebox{-2pt}{\includegraphics[height=1em]{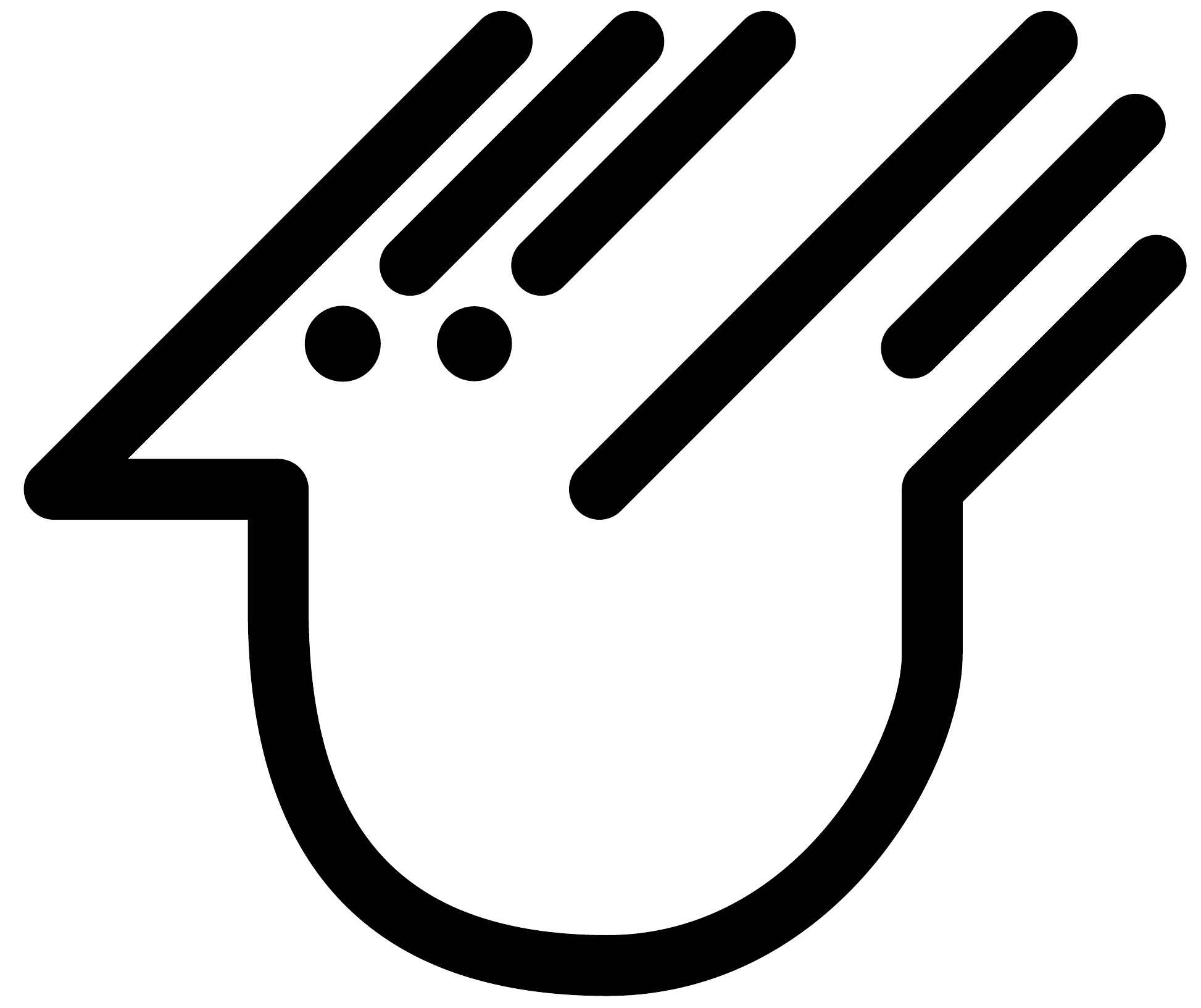}}%
}
\newcommand{\agdalogo}{%
  \usebox{\logoagdabox}}%

\newcommand{\zenodolink}{https://zenodo.org/records/15049780}
\newcommand{\Agda}{\href{\zenodolink}{\agdalogo}}

\newif\ifdraft
\draftfalse

\renewcommand{\max}[1]{\ifdraft{\color{blue}[{\bf Max says}: #1]}\fi}
\newcommand{\pedro}[1]{\ifdraft{\color{red}[{\bf Pedro says}: #1]}\fi}

\newcommand{\pipe}{\,|\,}

\newcommand{\states}{\mathtt{states}}

\newcommand{\transitions}{\texttt{transitions}}
\newcommand{\epstransitions}{\epsilon\texttt{transitions}}

\newcommand{\init}{\mathtt{init}}

\newcommand{\pparse}{\mathtt{parse}}
\newcommand{\print}{\mathtt{print}}

\newcommand{\Dyck}{\mathtt{Dyck}}
\newcommand{\N}{\texttt{N}}

\declaretheoremstyle[headfont=\normalfont\itshape]{defstyle}
\declaretheoremstyle[headfont=\normalfont\scshape]{thmstyle}

\declaretheorem[name=Definition, style=defstyle, numberwithin=section]{definition}
\declaretheorem[style=thmstyle, name=Theorem, numberlike=definition]{theorem}
\declaretheorem[style=thmstyle, name=Axiom, numberlike=definition]{axiom}
\declaretheorem[style=thmstyle, name=Corollary, numberlike=definition]{corollary}
\declaretheorem[style=thmstyle, name=Construction, numberlike=definition]{construction}

\crefname{definition}{Definition}{Definitions}
\crefname{theorem}{Theorem}{Theorems}
\crefname{axiom}{Axiom}{Axioms}
\crefname{corollary}{Corollary}{Corollarys}
\crefname{construction}{Construction}{Constructions}

\begin{document}

\pagebreak

\ifarxiv
\title{Intrinsic Verification of Parsers and Formal Grammar Theory in Dependent
  Lambek Calculus (Extended Version)}
\else
\title{Intrinsic Verification of Parsers and Formal Grammar Theory in Dependent Lambek Calculus}
\fi

\author{Steven Schaefer}
\orcid{0009-0007-1258-9501}
\affiliation{%
  \institution{University of Michigan}
  \city{Ann Arbor}
  \country{USA}
}
\email{stschaef@umich.edu}

\author{Nathan Varner}
\orcid{0009-0000-3031-4930}
\affiliation{%
  \institution{University of Michigan}
  \city{Ann Arbor}
  \country{USA}
}
\email{nmvarner@umich.edu}

\author{Pedro Henrique Azevedo de Amorim}
\orcid{0000-0002-8338-8973}
\affiliation{%
  \institution{University of Oxford}
  \city{Oxford}
  \country{United Kingdom}
}
\email{pedro.azevedodeamorim@cs.ox.ac.uk}

\author{Max S. New}
\orcid{0000-0001-8141-195X}
\affiliation{%
  \institution{University of Michigan}
  \city{Ann Arbor}
  \country{USA}
}
\email{maxsnew@umich.edu}

\makeatletter
\let\@authorsaddresses\@empty
\makeatother

\begin{abstract}
  We present \theoryname~(\theoryabbv), a domain-specific dependent
  type theory for verified parsing and formal grammar theory. In
  \theoryabbv, linear types are used as a syntax for formal grammars,
  and parsers can be written as linear terms. The linear typing
  restriction provides a form of intrinsic verification that a parser
  yields only valid parse trees for the input string. We demonstrate
  the expressivity of this system by showing that the combination of
  inductive linear types and dependency on non-linear data can be used
  to encode commonly used grammar formalisms such as regular and
  context-free grammars as well as traces of various types of
  automata. Using these encodings, we define parsers for regular
  expressions using deterministic automata, as well as
  examples of verified parsers of context-free grammars.

  We present a denotational semantics of our type theory that interprets the
  linear types as functions from strings to sets of abstract parse trees and
  terms as parse transformers. Based on this denotational semantics, we have
  made a prototype implementation of \theoryabbv using a shallow embedding in
  the Agda proof assistant. All of our examples parsers have been implemented in
  this prototype implementation.
\end{abstract}


\ifarxiv
\else
\fi

\maketitle

\section{Introduction}
\label{sec:intro}
Parsing structured data from untrusted input is a ubiquitous task in
computing. Any formally verified software system that interacts with the outside
world must contain some parsing component, and unverified parsers can undermine the overall correctness theorem for the system.
For example, in an extensive experiment finding bugs in C compilers
\cite{yangFindingUnderstandingBugs}, an early version of CompCert, the formally
verified C compiler, was found to have bugs only in the then
unverified parsing component \cite{leroy_formal_2009}.
The correctness of the compiler only showed that the abstract syntax tree would be compiled correctly, but this is
not very useful if it did not correctly correspond to the actual
source program. Eventually, CompCert
was updated to use a verified parser that implements
an $\LR{1}$ grammar
\cite{jourdanValidatingLRParsers2012}.

It is entirely understandable from an engineering perspective
\emph{why} verified parsing was not part of the initial releases of
CompCert: parsing algorithms are a complex area,
featuring a variety of domain-specific formalisms such as context-free
grammars and various automata. These formalisms have little relation
to the main components of a verified compiler. For this reason, it is
advantageous for verified parsers to be implemented using a reusable
verified library, just as parser generators and regular expression
matchers have done for many decades in unverified software.

Prior approaches to verified parsing focus on verification of a
particular grammar formalism, such as non-left-recursive grammars or
$\LL{1}$ grammars
\cite{lasserCoStarVerifiedALL2021,EdelmannZippy2020,danielssonTotalParserCombinators2010}.
This has led to a series of isolated solutions, as each new grammar formalism is extended with its own independent
verified implementation.

In this work, we present the design of \theoryname (\theoryabbv), a
domain-specific language for formal verification of parsers. A key
property is that \theoryabbv is an \emph{extensible} framework for
verification of parsers in that it supports the definition of grammar
formalisms of unrestricted complexity. That is, \theoryabbv is not a
system for verifying \emph{one} type of grammar formalism, but instead
is a domain-specific language in which many grammar formalisms and
their verified parsers can be implemented. For example, \theoryabbv itself is
not a verified parser generator compiling regular expressions to
deterministic finite automata but is instead a domain specific
language in which we can write such a verified parser generator.

The design of \theoryabbv
is an extension of Joachim Lambek's \emph{syntactic calculus}
\cite{lambek58}. Lambek calculus is a grammar formalism equivalent
in expressive power to context-free grammars that in modern
terminology would be considered a kind of \emph{non-commutative linear
logic} --- a version of linear logic where the tensor product is not
commutative, reflecting the obvious property that the relative
ordering of characters is significant in parsing problems. We extend
non-commutative linear logic with two key components that increase its
power to support arbitrarily powerful grammar formalisms: inductive
linear\footnote{Throughout, we shall use ``linear types'' to refer to the
non-commutative linear types.} types, as well as dependency of linear types on non-linear
data. The resulting system has two kinds of types: non-linear types
which model sets and linear types which model formal
grammars. Crucially, the non-linear types and linear types are allowed
to be dependent on non-linear types, but not on linear types. This
combination has been used previously in the ``linear-non-linear
dependent'' type theory with \emph{commutative} linear logic to model
imperative programming \cite{krishnaswami_integrating_2015}.

The substructural nature of \theoryabbv is well-aligned with the
requirements intrinsic to parsing and the theory of formal
languages, where strings constitute a clear notion of resource
that cannot be duplicated, reordered, or dropped. Moreover, \theoryabbv ensures that parsers written in the calculus are \emph{correct-by-construction}. That is, our type system is rich
enough that typing derivations carry intrinsic proofs of parser correctness.
Parsers written in \theoryabbv take on a linear functional style, which makes
them familiar to write and amenable to compositional verification techniques.

To show the feasibility of our design, we have implemented \theoryabbv
as a shallowly embedded domain-specific language in the Cubical Agda
proof assistant \cite{VezzosiMortbergAbel2019}. We have implemented many example
grammars and parsers in our system including regular expressions,
non-deterministic and deterministic automata, as well as some example
$\LL{1}$ context-free grammars and parsers using stack-based
automata. Throughout this paper \Agda~ will mark
results that are mechanized in our Agda development and provide a link to their
implementation.

Our Agda prototype is based on a \emph{denotational semantics} of
\theoryabbv. The core idea of the denotational semantics stems from an observation of Elliott:
an abstract notion of formal grammar can be given by a ``proof-relevant'' predicate on strings \cite{elliottSymbolicAutomaticDifferentiation2021}. That is, a \emph{formal grammar} $A$ is a
function $\StringSem \to \Set$ such that for a string $w$, $A~w$ is the
set of ``proofs'' showing that $w$ belongs to the language recognized
by $A$. We show that all linear types in \theoryabbv can be so
interpreted as an abstract formal grammar in this sense, and that
linear terms can be interpreted as a kind of \emph{parse transformer}, a function that
takes a parse tree from one grammar to a parse tree in a different
grammar but over the same underlying string.

Our contributions are then:
\begin{itemize}
  \item The design of \theoryname (\theoryabbv): A dependent
    linear-non-linear type theory for building verified parsers, which
    extends prior work on dependent linear-non-linear type theory to
    support inductive linear types.
  \item A demonstration of how to encode many common grammar and parser formalisms
    (regular expressions, (non-)deterministic automata, context-free
    grammars) within our type theory.
  \item A prototype implementation of \theoryabbv in Agda with all
    examples mechanized.
  \item A denotational semantics for \theoryabbv that shows that the
    parsers are in fact verified to be correct, and the equational theory is sound.
\end{itemize}

This paper begins in \cref{sec:type-theory-examples} by studying small
example programs from \theoryabbv to build intuition.  From there, in
\cref{sec:tt} we provide the syntax, typing, and equational theory of
\theoryname. In \cref{sec:applications} we demonstrate the
applicability of \theoryabbv for relating familiar grammar and
automata formalisms as well as building concrete parsers.  Then in
\cref{sec:semantics-and-metatheory}, we give a denotational semantics
that makes precise the connection between \theoryabbv syntax and
formal grammars. Finally, in \Cref{sec:discussion} we discuss related
and future work.

\ifarxiv\else {Proofs and syntactic forms elided from this paper are presented
    in full in the extended version \cite{schaefer2025intrinsicverificationparsersformal}.}\fi

\section{\theoryname by Example}
\label{sec:type-theory-examples}
To gain intuition for working in \theoryabbv, we begin with some
illustrative examples drawn from the theory of formal languages. Each
of our examples will be defined for strings over the three character
alphabet $\Sigma = \{ \texttt{a} , \texttt{b}, \texttt{c} \}$.

\newcommand{\A}{\texttt{A}}
\newcommand{\B}{\texttt{B}}
\newcommand{\I}{\texttt{I}}
\newcommand{\f}{\texttt{f}}
\newcommand{\g}{\texttt{g}}
\renewcommand{\L}{\texttt{L}}
\renewcommand{\a}{\texttt{a}}
\renewcommand{\b}{\texttt{b}}
\renewcommand{\c}{\texttt{c}}
\newcommand{\w}{\texttt{w}}

\paragraph{Finite Grammars}
First consider finite grammars --- those built from base types via disjunctions and
concatenations. The base types comprise characters drawn from the alphabet, the
empty string, and the empty grammar.
For each character $a$ in the alphabet we have a type $\literal a$ which
has a single parse tree for the string
\stringquote{a} and no parse trees at any other strings. The grammar $\I$ has a single
parse tree for the empty string $\epsilon = \stringquote{}$ and no parses for any other strings.
The final base type, the empty grammar $0$, has no parses for any string. We
use type-theoretic syntax to represent disjunction $\oplus$ and concatenation
$\otimes$ of
grammars. Over an input string $\w$, a parse of the disjunction $\A \oplus \B$ is either
a parse of $\A$ over the string $\w$ or a
parse of $\B$ over the string $\w$, along with a tag \texttt{inl} or \texttt{inr} indicating which case was taken. A parse of $\A \otimes \B$ for
$\w$ is a splitting of $w$ into two strings $\w_{\A}$ and $\w_{\B}$ with
parses for $\A$
and $\B$, respectively.

A derivation of a word $w$ in a grammar $A$ is given by a term in our calculus that satisfies the typing
$\internalize{w} \vdash e : A$,
where $\internalize \w$ is a context with one variable for each character of
$\w$. The term $e : A$ represents a \emph{parse tree} of $w$ for the grammar $A$. For
example, to define a parse tree for \stringquote{ab}, we use
the context $\internalize{\stringquote{ab}} = a : \literal a , b :
\literal b$. In \Cref{fig:fingram}, we give a lambda term and its
typing derivation to define a parse for a finite grammar.

\begin{figure}
\begin{floatlisting}
f : (*@$\uparrowcode$@*)('(*@\color{alphabetcolor}a@*)' (*@$\otimes$@*) '(*@\color{alphabetcolor}b@*)' (*@$\lto$@*) ('(*@\color{alphabetcolor}a@*)' (*@$\otimes$@*) '(*@\color{alphabetcolor}b@*)') (*@$\oplus$@*) '(*@\color{alphabetcolor}c@*)')
f (a , b) = inl (a , b)
\end{floatlisting}\begin{mathpar}
  \footnotesize
  \inferrule
  {
    \inferrule
    {
      \inferrule
      {~}
      {a : \literal a \vdash a : \literal a}
      \\
      \inferrule
      {~}
      {b : \literal b \vdash b : \literal b}
    }
    {a : \literal a , b : \literal b \vdash (a , b) : \literal a \otimes
      \literal b}
  }
  {a : \literal a , b : \literal b \vdash \texttt{f} := \inl(a , b) :
    (\literal a \otimes \literal b) \oplus \literal c}
\end{mathpar}
\caption{(\Agda) \stringquote{ab} is parsed by $(\literal a \otimes \literal b) \oplus \literal c$
}
\label{fig:fingram}
\end{figure}
For this interpretation of parse trees as terms to make sense, our
calculus cannot allow for \emph{any} of the usual structural rules of
type theory: weakening, contraction and exchange. Weakening allows
for variables to go unused, which would allow for a character in an input string to be ignored, yielding the erroneous parse tree \( {a : \literal a , b : \literal b \not \vdash a : \literal a} \).
Contraction allows for the same variable to be used twice, yielding the erroneous parse
\({a : \literal a , b : \literal b \not \vdash (a, a) : \literal a \otimes \literal a}
\).
Finally, the ordering of characters in a string cannot be ignored while
parsing, so we omit the exchange rule because it would allow
for variables in the context to be reordered. This prevents the erroneous derivation,
\(
  {a : \literal a , b : \literal b \not \vdash (b , a) : \literal b \otimes \literal a}
\).

\paragraph{Regular Expressions}
Regular expressions can be encoded as types generated by base types,
$\oplus$, and $\otimes$, and the Kleene star $(\cdot)^{\ast}$.  For a
grammar $\A$, we define the Kleene star $\A^{*}$ as a particular
\emph{inductive linear type} of linear lists, as shown in
\cref{fig:kleenestarinductive}. Here $\A^{*} : \L$ means we are defining $\A^*$
to be a \emph{linear} type. $\A^{*}$ has two constructors: $\nil$, which
builds a parse of type $\A^{*}$ from nothing; and $\cons$, which
linearly consumes a parse of $\A$ and a parse of $\A^{*}$ and builds a
parse of $\A^{*}$. This linear consumption is defined by the linear
function type $\lto$. The linear function type $\A \lto \B$ defines functions that
take in parses of $\A$ as input, \emph{consume} the input, and return a parse of
$\B$ as output. The arrow, $\uparrow$, wrapping these
constructors means that the constructors are not consumed
upon usage, and so are \emph{non-linear} values themselves, which are not small. That is, the
names $\nil$ and $\cons$ are function symbols that may be reused as many times
as we wish.

\begin{figure}
\begin{lstlisting}
data A(*@$^*$@*) : L where
  nil : (*@$\uparrowcode$@*)(A(*@$^*$@*))
  cons : (*@$\uparrowcode$@*)(A (*@$\lto$@*) A(*@$^*$@*) (*@$\lto$@*) A(*@$^*$@*))
\end{lstlisting}
\caption{Kleene Star as an inductive type}
\label{fig:kleenestarinductive}
\end{figure}

Through repeated application of the Kleene star constructors,
\cref{fig:kleenestarderivation} gives a derivation that shows
\stringquote{ab} is parsed by the regular expression $({\literal a}^{*}
\otimes \literal b) \oplus \literal c$. The leaves of the proof tree
that mention the arrow $\uparrow$ describe a cast from a non-linear
type to a linear type.  For instance, the premise of the leaf
involving $\nil$ views $\nil : \ltonl {({\literal a}^{*})}$ as the
name of a constructor, and a constructor should be nonlinearly valued
because we may call it several times (or not at all). However, the
conclusion of this leaf views $\nil : {\literal a}^{*}$ as a linear
value, which in our syntax is an implicit coercion from a nonlinear
value to a linear value. After we call the constructor it ``returns''
a value that may only be used a single time.

\begin{figure}
\begin{mathpar}
\footnotesize
\inferrule
{
  \inferrule
  {
    \inferrule
    {
      \inferrule
      {
        \inferrule
        {
          \inferrule
          {~}
          {\cdot \vdash \cons : \ltonl {({\literal a} \lto {\literal a}^* \lto {\literal a}^*)}}
        }
        {\cdot \vdash \cons : {\literal a} \lto {\literal a}^* \lto {\literal a}^*}
        \\
        \inferrule
        {~}
        {a : {\literal a} \vdash a : {\literal a}}
      }
      {a : a \vdash \cons~a : {\literal a}^* \lto {\literal a}^*}
      \\
      \inferrule
      {
        \inferrule
        {~}
        {\cdot \vdash \nil : \ltonl {({\literal a}^*)}}
      }
      {\cdot \vdash \nil : {\literal a}^*}
    }
    {a : {\literal a} \vdash \cons~a~\nil : {\literal a}^{*}}
    \\
    \inferrule
    {~}
    {b : {\literal b} \vdash b : {\literal b}}
  }
  {a : {\literal a} , b : {\literal b} \vdash (\cons~a~\nil , b) : {\literal a}^{*} \otimes {\literal b}}
}
{a : \literal a , b : \literal b \vdash \texttt{g} := \inl (\cons~a~\nil , b) : ({\literal a}^{*} \otimes \literal b) \oplus \literal c}
\end{mathpar}
\begin{lstlisting}
g : (*@$\uparrowcode$@*)(('(*@\color{alphabetcolor}a@*)' (*@$\otimes$@*) '(*@\color{alphabetcolor}b@*)') (*@$\lto$@*) ('(*@\color{alphabetcolor}a@*)'(*@$^*$@*) (*@$\otimes$@*) '(*@\color{alphabetcolor}b@*)') (*@$\oplus$@*) '(*@\color{alphabetcolor}c@*)')
g (a , b) = inl (cons a nil , b)
\end{lstlisting}
\caption{(\Agda) \stringquote{ab} is parsed by
  $({\literal a}^{*} \otimes \literal b) \oplus \literal c$}
\label{fig:kleenestarderivation}
\end{figure}

We may also have derivations where the term in context is not simply a
string of literals. In \cref{fig:kleeneabstractproof} we show that every parse
of the grammar $(A \otimes A)^{*}$ induces a parse of $A^{*}$ for an
arbitrary grammar $A$. The context $(A \otimes A)^{*}$ does not correspond directly to a string, so it is
not quite appropriate
to think of a linear term here as a parse \textit{tree}.
The
context $a : (A \otimes A)^{*}$ does not contain concrete data to be parsed; rather, there may be many choices of string underlying the parse tree captured
by the variable $a$. Thus, the term $\texttt{h}$ from
\cref{fig:kleeneabstractproof} is not a parse of a string, and
it is more
appropriate to think of it as a parse \textit{transformer} --- a function from
parses of $(\A \otimes \A)^{*}$ to parses of $\A^{*}$.

We define $h$ by recursion on terms of type $(\A \otimes \A)^{*}$.
This recursion is expressed in the derivation tree by invoking the
elimination principle for Kleene star, written
as $\fold$. The parse transformer $h$ is more intuitively presented in the pseudocode
of \cref{fig:kleeneabstractproof} by pattern matching on the input and making an
explicit recursive call in the body of its definition.

\begin{figure}
\footnotesize
\begin{mathpar}
  \inferrule
  {
    \inferrule
    {~}
    {\cdot \vdash \nil : \A^*}
    \\
    \inferrule
    {
      \inferrule
      {
        \inferrule
        {~}
        {a_1 : \A, a_2 : \A , as : \A^* \vdash \cons~a_1 (\cons~a_2~as) : \A^*}
      }
      {aa : \A \otimes \A , as : \A^* \vdash \letin {(a_1 , a_2)}
        {aa} {\cons~a_1 (\cons~a_2~as)} : \A^*}
    }
    {\cdot \vdash f := \lamblto {aa} {\lamblto {as} {\letin {(a_1 , a_2)}
        {aa} {\cons~a_1 (\cons~a_2~as)}}} : (\A \otimes \A) \lto \A^* \lto \A^*}
  }
  {aas : (\A \otimes \A)^* \vdash \texttt{h} := \fold(\nil , f)(aas) : \A^* }
\end{mathpar}
\begin{lstlisting}
h : (*@$\uparrowcode$@*)((A (*@$\otimes$@*) A)(*@$^*$@*) (*@$\lto$@*) A(*@$^*$@*))
h nil = nil
h (cons (a1 , a2) as) = cons a1 (cons a2 (h as))
\end{lstlisting}
\caption{(\Agda) A parse transformer for abstract grammars}
\label{fig:kleeneabstractproof}
\end{figure}

\paragraph{Non-deterministic Finite Automata}
\newcommand{\s}{\texttt{s}}
\newcommand{\0}{\texttt{0}}
\newcommand{\1}{\texttt{1}}
\newcommand{\2}{\texttt{2}}
Regular expressions are a compact formalism for defining a formal grammar,
but an expression such as
$({\literal a}^{*} \otimes \literal b) \oplus \literal c$ does not give an
operationalized method for parsing. For this reason, most parsers are
implemented by
compiling a grammar to a corresponding automaton, which is readily
implemented. To implement
these algorithms in \theoryabbv, we represent automata as types, just as we did with regular expressions.

Finite automata are precisely the class of machines that recognize regular
expressions. \cref{fig:exampleNFA} shows a non-deterministic
finite automaton (NFA) for the regular expression
$({\literal a}^{*} \otimes \literal b) \oplus \literal c$, along with a type $\Trace$, an \emph{indexed} inductive linear type of traces through this automaton. Defining an indexed inductive type can be thought
of as defining a family of mutually recursive inductive types, one for each element of the indexing type. Here $\Trace$
uses an index $\s : \texttt{Fin 3}$ which picks out which
state in the automaton a trace begins at --- where $\texttt{Fin 3}$ is the
finite type containing inhabitants $\{\0 , \1 , \2\}$. We can think of this as defining three mutually recursive inductive types $\Trace~\0$,
$\Trace~\1$, and $\Trace~\2$.
There are three kinds of constructors for $\Trace$: (1) those that
terminate traces, (2) those that correspond to transitions labeled by
a character, and (3) those that correspond to transitions labeled by
the empty string $\epsilon$. The constructor $\texttt{stop}$
terminates a trace in the accepting state $\2$. The constructors
$\texttt{1to1}$, $\texttt{1to2}$, $\texttt{0to2}$ each define a
labeled transition through the NFA, and each of these consumes a parse
of the label's character and a trace beginning at the destination of a
transition to produce a trace beginning at the source of a
transition. The constructor $\texttt{0to1}$ behaves similarly, except
its transition is labeled with the empty string $\epsilon$. Therefore,
$\texttt{0to1}$ takes in a trace beginning at state $\1$ and returns a
trace beginning at state $\0$ corresponding to the same underlying
string.
Lastly, we give a $\lambda$ term that constructs an accepting trace
starting at the initial state for the string \stringquote{ab}.  Later
in \Cref{sec:applications}, we will show that we can actually
construct mutually inverse functions between the regular expression
$({\literal a}^{*} \otimes \literal b) \oplus \literal c$ and its
corresponding NFA traces ($\Trace~\0$) demonstrating that the regular
expression and the automaton capture the same language. Further, since
the functions are mutually inverse, this shows they are \emph{strongly
equivalent} as grammars.

\begin{figure}
  \begin{minipage}[t]{.6\textwidth}
    \vspace{0pt}
  \begin{floatlisting}
data Trace : (s : Fin 3) (*@$\tocode$@*) L where
  stop : (*@$\uparrowcode$@*)(Trace 2)
  1to1 : (*@$\uparrowcode$@*)('(*@\color{alphabetcolor}a@*)' (*@$\lto$@*) Trace 1 (*@$\lto$@*) Trace 1)
  1to2 : (*@$\uparrowcode$@*)('(*@\color{alphabetcolor}b@*)' (*@$\lto$@*) Trace 2 (*@$\lto$@*) Trace 1)
  0to2 : (*@$\uparrowcode$@*)('(*@\color{alphabetcolor}c@*)' (*@$\lto$@*) Trace 2 (*@$\lto$@*) Trace 0)
  0to1 : (*@$\uparrowcode$@*)(Trace 1 (*@$\lto$@*) Trace 0)

k : (*@$\uparrowcode$@*)(('(*@\color{alphabetcolor}a@*)' (*@$\otimes$@*) '(*@\color{alphabetcolor}b@*)') (*@$\lto$@*)  Trace 0)
k (a , b) = 0to1 (1to1 a (1to2 b stop))
\end{floatlisting}
  \end{minipage}%
  \begin{minipage}[t]{.4\textwidth}
  \vspace{8pt}
  \hspace{10pt}
  \begin{tikzpicture}[node distance = 25mm ]
    \node[state, initial] (0) {$\0$};
    \node[state, below left of=0] (1) {$\1$};
    \node[state, right of=1, accepting] (2) {$\2$};

    \path[->] (0) edge[above] node{$\epsilon$} (1)
              (0) edge[right] node{$\stringquote{c}$} (2)
              (1) edge[loop left] node{$\stringquote{a}$} (1)
              (1) edge[below] node{$\stringquote{b}$} (2);
  \end{tikzpicture}
  \end{minipage}%
  \caption{(\Agda) NFA for $(\a^{*} \otimes \b) \oplus \c$ and its corresponding type}
  \label{fig:exampleNFA}
\end{figure}
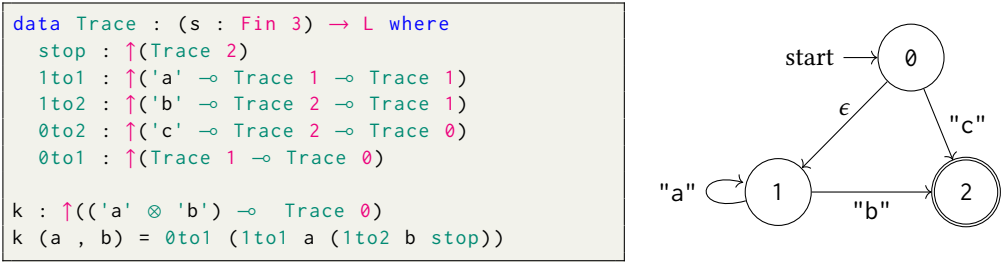

\section{Syntax and Typing for \theoryname}
\label{sec:tt}

The design of \theoryabbv is based on the dependent
linear-non-linear calculus (\lnld) and Lambek calculus, also known as
non-commutative linear logic
\cite{krishnaswami_integrating_2015,lambek58}. As in \lnld,
\theoryabbv includes both non-linear dependent types, as well as
linear types, which are allowed to depend on the non-linear types, but
not on other linear types.
The main point of departure from \lnld's design is that, as in Lambek calculus \cite{lambek58}, the linear
typing is \emph{non-commutative} --- i.e., that exchange is not an
admissible structural rule. Furthermore, we add a general-purpose
indexed inductive linear type connective, as well as an
\emph{equalizer} type, which we will show allows us to perform
inductive proofs of equalities between linear terms.
Finally, while \lnld was enhanced with special connectives
inspired by separation logic to model imperative programming, we
instead add base types and axioms to the system specifically to model
formal grammars and parsing.

The formation rules for the judgments of \theoryabbv are shown in
\Cref{fig:contexts,fig:non-linear-types,fig:linear-formation}. $\Gamma$ stands for non-linear contexts; $X,Y,Z$
stand for non-linear types; $M,N$ stand for non-linear terms, these
act as in an ordinary dependent type theory; $\Delta$ stands for
linear contexts; $A,B,C$ for linear types; and, $e,f,g$ for linear
terms. These contexts, types and terms are allowed to depend on an
ambient non-linear context $\Gamma$, but note that linear types $A$
cannot depend on any \emph{linear} variables in $\Delta$. We include
definitional equality judgments for both kinds of type and term
judgments as well. Additionally, we have judgments $\Gamma \vdash X
\isSmall$ and $\Gamma \vdash A \isSmallLin$ which are used in the definition of
universe types.

\begin{figure}
  \begin{mathpar}
    \footnotesize
    \boxed{\ctxwff \Gamma}

    \inferrule{~}{\ctxwff \cdot}

    \inferrule{\ctxwff \Gamma \and \ctxwffjdg \Gamma X}{\ctxwff {\Gamma , x : X}}

    \boxed{\linctxwff \Gamma \Delta}

    \inferrule{~}{\linctxwff \Gamma \cdot}

    \inferrule{\linctxwff \Gamma \Delta \and \linctxwffjdg \Gamma A}{\linctxwff
      {\Gamma} {\Delta , a : A}}
  \end{mathpar}
  \caption{Context well-formedness rules}
  \label{fig:contexts}
\end{figure}

\subsection{Non-linear Typing}

\begin{figure}
  \begin{mathpar}
  \footnotesize
    \boxed{\inferrule{\ctxwff\Gamma}{\ctxwffjdg \Gamma X}}

    \inferrule{~}{\ctxwffjdg \Gamma {U}}

    \inferrule{~}{\ctxwffjdg \Gamma {L}}

    \inferrule{~}{\ctxwffjdg \Gamma 1}

    \inferrule{~}{\ctxwffjdg \Gamma \bot}

    \inferrule{~}{\ctxwffjdg \Gamma {Bool}}

    \inferrule{~}{\ctxwffjdg \Gamma {Nat}}

    \inferrule{\ctxwffjdg \Gamma X \and \ctxwffjdg {\Gamma , x : X} {Y}}{\ctxwffjdg \Gamma {\PiTyLimit {x}{X}{Y}}}

    \inferrule{\ctxwffjdg \Gamma X \and \ctxwffjdg {\Gamma , x : X} {Y}}{\ctxwffjdg \Gamma {\SigTyLimit {x}{X}{Y}}}

    \\

    \inferrule{\nonlinterm \Gamma M U}{\ctxwffjdg \Gamma {\unquoteTy M}}

    \inferrule{\linctxwffjdg \Gamma A}{\ctxwffjdg \Gamma {\ltonl A}}

    \inferrule{\ctxwffjdg \Gamma X \and \nonlinterm \Gamma M X \and \nonlinterm \Gamma N X}{\ctxwffjdg \Gamma {M =_{X} N}}

    \\

    \boxed{\inferrule{\ctxwffjdg\Gamma X \and \ctxwffjdg\Gamma Y}{\ctxwffjdg \Gamma {X \equiv Y}}}

    \inferrule{\nonlinterm \Gamma {X \equiv Y} {U}}{\ctxwffjdg \Gamma {X \equiv Y}}

    \inferrule{\linctxwffjdg \Gamma X}{\linctxwffjdg \Gamma {\unquoteTy
        {\quoteTy X} \equiv X}}

    \boxed{\inferrule{\ctxwffjdg \Gamma X}{\Gamma \vdash X \isSmall}}

    \\
    \boxed{\inferrule{\ctxwffjdg \Gamma X}{\nonlinterm \Gamma M X}}

    \inferrule{\Gamma \vdash  X \isSmall}{\nonlinterm \Gamma {\quoteTy X} U}

    \inferrule{\Gamma \vdash A \isSmallLin}{\nonlinterm \Gamma {\quoteTy A} L}

    \inferrule{\linterm \Gamma \cdot e A}{\nonlinterm \Gamma e \ltonl{A}}

    \inferrule{\nonlinterm \Gamma {M \equiv N} {X}}
      {\nonlinterm \Gamma {refl} {M =_{X} N}}

    \\
    \boxed{\inferrule{\nonlinterm\Gamma M X \and \nonlinterm \Gamma N X}{\nonlinterm \Gamma {M \equiv N} X}}

    \inferrule{\nonlinterm \Gamma P {M =_{X} N}}{\nonlinterm\Gamma{M \equiv N} X}
  \end{mathpar}
  \caption{Non-linear Formation and Typing Rules (selection)}
  \label{fig:non-linear-types}
\end{figure}

We present a selection of the non-linear type constructors in
\Cref{fig:non-linear-types} and provide the rest in
\ifarxiv{\cref{fig:full-non-linear-types,fig:jdg-eq-nonlinear} of
    \cref{sec:syntax}}\else{the extended version of the paper}\fi.  First,
we include universe types $U$ of small non-linear types and $L$ of linear
types. These are defined as
universes ``ala Coquand'' in that we define judgments saying when
non-linear and linear types are \emph{small} and define the universes to internalize
precisely this judgment \cite{coquandPresheafModel,lmcs:7713}. The definition of
smallness is simply that all types are small as long as their sub-formulae are,
with the exception of the two universe types themselves. A formal
description of smallness is given in
\ifarxiv{\cref{fig:small-nonlin-ty,fig:small-lin-ty} of
    \cref{sec:syntax}}\else{the extended version of the paper}\fi. These universe
types are needed so that we can define non-linear and linear types by recursion on natural
numbers.
Next, we include standard $\Sigma,\Pi$, empty, unit, Boolean, and natural number
types. More complex inductive types --- such as sum types, list types, and
$Fin~n$ --- can be defined in terms of these primitives, and we give their
encoding in \ifarxiv{\cref{sec:syntax}}\else{the extended version of the paper}\fi.

We use an \emph{extensional} equality type $M =_{X} N$ with introduction form
$refl$, but no elimination form. Instead we have the equality reflection rule
which allows us to conclude a definitional equality $M =_{X} N$ from an
arbitrary typal equality proof $P$.  The usage of an extensional equality type
matches our implementation, which interprets both judgmental and typal equality
as Cubical Agda's \texttt{Path} type, and so naturally supports the equality
reflection rule. Extensional equality makes type checking of our syntax as such
undecidable \citep{Hofmann_1997} because the conversion rule may require an arbitrarily complex
equality proof with no explicit proof term. However, in our Agda implementation we
must provide all of these equalities manually, so the extensionality does
not raise any issues. This makes
\theoryabbv into an extensional theory, supporting function extensionality and
the uniqueness of identity proofs. The development could be ported to an intensional type theory in the future, possibly requiring the use of setoids to handle function extensionality.

Lastly, we include a non-linear type $\ltonl A$ where $A$ is a linear type. The
intuition for this type is that its elements are the linear terms that are
``resource free'': its introduction rule says we can construct an $\ltonl A$
when we have a linear term of type $A$ with no free linear
variables. Semantically, this is the type of parses of the empty string. This
type is used extensively in our examples, playing a similar role to the $!$
modality of ordinary linear logic or the persistence modality $\square$ of
separation logic \cite{girard_linear_1987,jung_higher-order_2016}.

\subsection{Linear Typing}

\begin{figure}
  \begin{mathpar}
    \footnotesize
    \boxed{\inferrule{\ctxwff\Gamma}{\linctxwffjdg \Gamma A}}

    \inferrule{~}{\linctxwffjdg \Gamma I}

    \inferrule{c \in \Sigma}{\linctxwffjdg \Gamma {\literal{c}}}

    \inferrule{\linctxwffjdg \Gamma A \and \linctxwffjdg \Gamma B}{\linctxwffjdg
      \Gamma {A \otimes B}}

    \inferrule{\linctxwffjdg \Gamma A \and \linctxwffjdg \Gamma B}{\linctxwffjdg
      \Gamma {A \lto B}}

    \inferrule{\linctxwffjdg \Gamma A \and \linctxwffjdg \Gamma B}{\linctxwffjdg
      \Gamma {A \tol B}}

    \inferrule{\linctxwffjdg {\Gamma, x : X} A}{\linctxwffjdg
      \Gamma {\LinSigTyLimit{x}{X}{A}}}

    \inferrule{\linctxwffjdg {\Gamma, x : X} A}{\linctxwffjdg
      \Gamma {\LinPiTyLimit{x}{X}{A}}}

    \inferrule{\nonlinterm \Gamma f {\ltonl{(A \lto B)}} \and
      \nonlinterm \Gamma g {\ltonl{(A \lto B)}}}{\linctxwffjdg \Gamma {\equalizer {a}{f}{g}}}

    \inferrule{\nonlinterm \Gamma M L}{\linctxwffjdg \Gamma {\unquoteTy M}}

    \boxed{\inferrule{\linctxwffjdg \Gamma A}{\Gamma \vdash A \isSmallLin}}

    \\
    \boxed{\inferrule{\linctxwffjdg \Gamma A \and \linctxwffjdg \Gamma B}{\linctxwffjdg \Gamma {A \equiv B}}}

    \inferrule{\nonlinterm \Gamma {A \equiv B} L}{\linctxwffjdg
      \Gamma {A \equiv B}}

    \inferrule{\linctxwffjdg \Gamma A}{\linctxwffjdg \Gamma {\unquoteTy
        {\quoteTy A} \equiv A}}
  \end{mathpar}
  \caption{Linear Type Formers, Type Equivalence}
  \label{fig:linear-formation}
\end{figure}

We give the rules for linear type
formation in \cref{fig:linear-formation} and the definition of linear terms in
\cref{fig:linear-terms}. The equational theory for these
types is straightforward $\beta\eta$ equivalence and is included in
\ifarxiv{\cref{fig:jdgeq} of \cref{sec:denotational}}\else{the extended version of the paper}\fi.

First, the
linear variable rule says that a linear variable
can be used if it is the \emph{only} variable in the context. Next, we cover the
``multiplicative'' connectives of non-commutative
linear logic. The linear unit ($\I$) and tensor product ($\otimes$)
are standard for a non-commutative linear logic: when we construct a
linear unit we cannot use any variables and when we construct a tensor
product, the two sides must use disjoint variables, and the variables
the left side of the product uses must be to the left in the context
of the variables used by the right side of the tensor product. The
elimination rules for unit and tensor are given by pattern
matching. The pattern matching rules split the linear context into
three pieces $\Delta_1,\Delta_2,\Delta_3$: the middle $\Delta_2$ is
used by the scrutinee of the pattern match, and in the continuation
this context is replaced by the variables brought into scope by the
pattern match. This ensures that pattern matches maintain the proper
ordering of resource usage.

Because we are non-commutative, there are two function types: $A \lto
B$ and $B \tol A$, which have similar $\lambda$ introduction forms and
application elimination forms. The difference between these is that
the introduction rule for $A \lto B$ adds a variable to the right side
of the context, whereas the introduction rule for $B \tol A$ adds a
variable to the left side of the context. In our experience, because
by convention parsing algorithms parse from left-to-right, we rarely
need to use the $B \tol A$ connective. As we have already seen, the
$\lto$ connective is frequently used in conjunction with the
$\uparrow$ connective so that we can abstract non-linearly over linear
functions.

Next, we cover the ``additive'' connectives. First, we use the
non-linear types to define \emph{indexed} versions of the additive
disjunction $\oplus$ and additive conjunction $\&$ of linear logic,
which can be thought of as linear versions of the $\Sigma$ and $\Pi$
connectives of ordinary dependent type theory, respectively. The
indexed $\&$ is defined by a $\lambda$ that brings a \emph{non-linear}
variable into scope and eliminated using projection where the index
specified is given by a non-linear term. The rules for indexed
$\oplus$ are analogous to a ``\emph{weak}'' $\Sigma$ type: it has an
injection introduction rule $\sigma$, but its elimination rule is
given by \emph{pattern matching} rather than first and second
projections. We can define the more typical nullary and binary
versions of these connectives by using indexing over the empty and
boolean type respectively. We will freely use $0$ to refer to this
empty disjunction and $\top$ to refer to the empty conjunction, and
use infix $\oplus/\&$ for binary disjunction/conjunction.

Lastly, we include a type $\equalizer {a}{f}{g}$ that we call the
\emph{equalizer} of linear functions $f$ and $g$. We think of this
type as the ``subtype'' of elements of $A$ that satisfy the equation
$f\,a\equiv g\,a$. Note that it is important here that $f, g$
themselves are non-linearly used functions, as linear values cannot be
used in a type.  Equalizer types are not needed for non-linear types
since they can be constructed using the equality type as $\sum_{x:X}
f\,x=_Y g\,x$, but this construction cannot be used for linear types
because it uses a \emph{dependent} version of the equality type, which
we cannot define as a linear type. While the equalizer type is not
used directly in defining any of our parsers or formal grammars, it is
used for several proofs, allowing for inductive arguments about our
indexed inductive types.

In addition to these type-theoretic principles, we need two additional
axioms that do not generally hold in systems based on linear
logic. First, we need that additive conjunction \emph{distributes}
over additive disjunction --- e.g., in the finitary case that $0 \& A
\cong 0$ and $(A + B) \& C \cong (A \& C) + (B \& C)$. More generally, we assume
\cref{ax:dist}.

\begin{axiom}[Distributivity]
\label{ax:dist}
For any $A : (x : X) \to Y(x) \to L$, the definable function distributing
conjunction over disjunction
\(\LinSigTyLimit{f}{\PiTy{x}{X}{Y(x)}}{\LinPiTyLimit{x}{X}{A\,x\,(f\,x)}} \lto
\LinPiTyLimit{x}{X}{\LinSigTyLimit{y}{Y(x)}{A\,x\,y}}
\)
has an inverse.
\end{axiom}

The following corollary is a well known consequence of distributivity \cite{Cockett_1993}, which we use in \cref{lem:unambig-to-disjoint} to prove that unambiguous binary sums have unambiguous summands.
\begin{corollary}
  \label{cor:binary-mono}
  Distributivity implies that the constructors $\inl : A \to A \oplus B$,
  $\inr : B \to A \oplus B$ of a binary sum are injective --- i.e. if
  $\inl\,a \equiv \inl\,a'$, then $a \equiv a'$.
\end{corollary}

Our primary use of distributivity is to define an equivalence that
expresses a linear type $A$ as a sum over which character it starts with, if any,
\(
A \cong \left(A \& I \right) \oplus \LinSigTy{c}{\Sigma_0}{\left( A \& \left( \literal{c} \otimes \top \right) \right)}
\). We use this equivalence when building a parser for the traces of the lookahead
automaton in \cref{fig:binop-inductive}.

Second, we need that the different constructors of
$\bigoplus$ are disjoint. That is, we want to enforce that $\sigma~x~a$ and
$\sigma~x'~a'$ of type $\LinSigTy{x}{X}{A\,x}$ are not equal whenever $x \neq x'$.
However, because these are linear terms we cannot state their disequality
directly. Instead, we encode the disequality via a function out of an equalizer,

\begin{axiom}[$\sigma$-Disjointness]
\label{ax:disjointness}
For any $A : X \to L$ and $x \neq x' : X$ there is a function,
\[
  \uparrow(\equalizer{b}{\left( \sigma\,x \circ \pi_1 \right)}{\left(  \sigma\,x'\circ \pi_2 \right)}
  \lto 0)
\]
  where $b: A(x) \& A(x')$. That is, the grammar of pairs of an
$a:A(x)$ and an $a': A(x')$ such that $\sigma\,x\,a = \sigma\,x'\,a'$
is empty.
\end{axiom}

We use $\sigma$-disjointness in \cref{lem:unambig-to-disjoint} to prove that
unambiguous sums have disjoint summands.

\begin{figure}
  \footnotesize
  \begin{mathpar}
    \boxed{
      \inferrule
      {\linctxwff \Gamma \Delta \and \linctxwffjdg \Gamma \A}
      {\linterm \Gamma \Delta {e} {A}}}

    \inferrule{~}{\Gamma ; a : A \vdash a : A}
    \and
    \inferrule{\nonlinterm \Gamma {M} {\ltonl A}}{\linterm {\Gamma} {\cdot} {M} {A}}
    \and
    \inferrule{\Gamma ; \Delta \vdash e : B \\ \linctxwffjdg \Gamma {A \equiv B}}{\Gamma ; \Delta \vdash e : A}
    \\
    \inferrule{~}{\Gamma ; \cdot \vdash () : I}
    \and
    \inferrule{\Gamma ; \Delta_2 \vdash e : I \\ \Gamma ; \Delta_1,\Delta_3 \vdash e' : C}{\Gamma ; \Delta_1,\Delta_2,\Delta_3 \vdash \letin {()} e {e'} : C}
    \\
    \inferrule{\Gamma ; \Delta \vdash e : A \\ \Gamma ; \Delta' \vdash e' : B}{\Gamma ; \Delta, \Delta' \vdash (e , e') : A \otimes B}
    \and
    \inferrule{\Gamma ; \Delta_2 \vdash e : A \otimes B \\ \Gamma ; \Delta_1, a : A, b : B, \Delta_2 \vdash e' : C}{\Gamma ;  \Delta_1, \Delta_2, \Delta_3 \vdash \letin {(a , b)} e {e'} : C}
    \\
    \inferrule{\Gamma ; \Delta , a : A \vdash e : B}{\Gamma ; \Delta \vdash \lamblto a e : A\lto B}
    \and
    \inferrule{\Gamma ; \Delta \vdash e : A \lto B \\ \Gamma ; \Delta' \vdash e' : A} {\Gamma ; \Delta, \Delta' \vdash \applto {e} {e'} : B}
    \\
    \inferrule{\Gamma ; a : A , \Delta \vdash e : B}{\Gamma ; \Delta \vdash \lambtol a e : B\tol A}
    \and
    \inferrule{\Gamma ; \Delta \vdash e : A \\ \Gamma ; \Delta' \vdash e'
      : B \tol A}{\Gamma ; \Delta, \Delta' \vdash \apptol {e'} {e} : B}
    \\
    \inferrule{\Gamma, x : X ; \Delta  \vdash e : A}
              {\Gamma ; \Delta \vdash \dlamb x e : \LinPiTy x X A}
    \and
    \inferrule{\Gamma ; \Delta \vdash e : \LinPiTy x X A \\ \Gamma \vdash M : X}{\Gamma ; \Delta \vdash e\,.\pi\,M : \subst A {M} x}
    \\
    \inferrule{\Gamma \vdash M : X \quad \Gamma ; \Delta \vdash e : \subst A M x}{\Gamma ; \Delta \vdash \sigma\,M\,e : \bigoplus\limits_{x:X} A}
    \and
    \inferrule{\Gamma ; \Delta_2 \vdash e : \bigoplus\limits_{x:X} A \quad \Gamma, x : X ; \Delta_1, a : A, \Delta_3 \vdash e' : C}{\Gamma; \Delta_1, \Delta_2, \Delta_3 \vdash \letin {\sigma\,x\,a} e {e'}: C}
    \\
    \inferrule
    { \Gamma ; \Delta \vdash e : A \\
      \Gamma ; \Delta \vdash \applto {f}{e} \equiv \applto {g}{e}}
    {\Gamma ; \Delta \vdash \equalizerin{e} : \equalizer {a}{f}{g}}
    \and
    \inferrule
    {\Gamma ; \Delta \vdash e : \equalizer{a}{f}{g}}
    {\Gamma ; \Delta \vdash \equalizerpi {e} : A}
  \end{mathpar}
  \caption{Linear terms}
  \label{fig:linear-terms}
\end{figure}

\subsection{Indexed Inductive Linear Types}

Next, we introduce the most complex and important linear type
constructors of our development, \emph{indexed inductive linear
types}. We encode these by adding a mechanism for constructing initial
algebras of strictly positive functorial type expressions, following
prior work on inductive types
\cite{nakov_quantitative_2022,altenkirch_indexed_2015}. The syntax is
given in \Cref{fig:iilt}. First, we add a non-linear type $\SPF\,X$ of
\emph{strictly positive functorial} linear type expressions indexed by
a non-linear type $X$. We think of the elements of this type as
syntactic descriptions of linear types that are parameterized by
$X$-many variables standing for linear types that are only used in
strictly positive positions. Accordingly, the $\SPF\,X$ type supports
an operation $\el$ that interprets it as such a type constructor, as
well as an operator $\map$ that defines a functorial action on parse
transformers. The $\SPF\,X$ type supports constructors for a reference
$\Var~x$ to one of the linear type variables, a constant expression
that does not mention any type variables $K$, as well as tensor
products and additive conjunction and disjunction of type expressions.
Further, we add equations in \ifarxiv{\cref{fig:spf-act} of \cref{sec:syntax}
  }\else{the extended version\fi that say that the
$\el$/$\map$ operations correspond to these descriptions of the
constructors.

\begin{figure}
  \begin{footnotesize}
    \begin{mathpar}
      \inferrule{\Gamma \vdash X\isTy}{\Gamma \vdash \SPF\,X \isTy}\and
      \inferrule{\Gamma \vdash X\isSmall}{\Gamma \vdash \SPF\,X \isSmall}\and
    \el : \prod_{X:U}\SPF\,X \to (X \to L) \to L\and
    \map : \prod_{X:U}\prod_{F : \SPF\,X}\prod_{A,B:X\to L}{\left(\prod_{x:X}\uparrow(\unquoteTy{A\,x}\lto \unquoteTy{B\,x})\right)} \to {\uparrow(\unquoteTy{\el(F)(A)} \lto \unquoteTy{\el(F)(B)})}\and
    \mathsf{Var} : \prod_{X:U} X \to \SPF\,X\and
    \mathsf{K} : \prod_{X:U} L \to \SPF\,X\and
    \mathsf{\bigoplus} : \prod_{X:U}\prod_{Y:U}(Y \to \SPF\,X) \to \SPF\,X\and
    \mathsf{\bigamp} : \prod_{X:U}\prod_{Y:U}(Y \to \SPF\,X) \to \SPF\,X\and
    \mathsf{\otimes} : \prod_{X:U}\SPF\,X \to \SPF\,X \to \SPF\,X\and
    \roll : \prod_{X:U}\prod_{F:X \to \SPF\,X}\prod_{x:X}\ltonl{(\el(F\,x)(\mu\,F))}\and
    \fold : \prod_{X:U}\prod_{F:X\to\SPF\,X}\prod_{A:X \to L}
    \left(\prod_{x:X}\ltonl{(\unquoteTy{\el(F\,x)(\unquoteTy{A})} \lto \unquoteTy{A\,x})}\right)
    \to \prod_{x:X}\ltonl{(\mu F\,x \lto A\,x)}\and

    \inferrule*[right=Ind$\beta$]
    {\Gamma \vdash f : \prod_{x:X}\ltonl{(\el(F\,x)(A) \lto A\,x)} \and \Gamma; \Delta \vdash e : \el(F\,x)(\mu\,F)}
    {\Gamma;\Delta \vdash \fold\,F\,f\,x\,(\roll\,e) \equiv f\,x\,(\map(F\,x)\,(\fold\,F\,f)) : A\,x}

    \inferrule*[right=Ind$\eta$]
    {\Gamma \vdash f : \prod_{x:X}\ltonl{(\el(F\,x)(A) \lto A\,x)}
     \and \Gamma \vdash e : \prod_{x:X}\ltonl{(\mu F\,x \lto A\,x)}\\\\
      \Gamma,x:X;a:\el(F\,x)(\mu F) \vdash e\,x\,(\roll\,a) \equiv f\,x\,(\map(F\,x)\,e) : A\,x }
    {\fold\,F\,f\equiv e' : \prod_{x:X}\ltonl{(\mu F\,x \lto A\,x)}}
  \end{mathpar}
  \end{footnotesize}
  \caption{Strictly positive functors and indexed inductive linear types}
  \label{fig:iilt}
\end{figure}

Next, given a family of $X$-many strictly positive linear type
expressions $F : X \to \SPF\,X$, we define a family $\mu F : X \to L$
of $X$-many mutually recursive inductive types. The introduction rule
for this is $\roll$, which constructs an element of $\mu F\,x$ from
the one-level of the $x$th type expression. The elimination principle
is defined by a mutual $\fold$ operation: given a family of output
types $A$ indexed by $X$, we can define a family of functions from
$\mu F\,x \multimap A\,x$ if you specify how to interpret all of the
constructors as operations on $A$ values. We add $\beta\eta$ equations
that specify that this makes the family $\mu F$ into an \emph{initial
algebra} for the functor $\el(F)$. That is, the $\beta$ rule says that
a $\mathsf{fold}$ applied to a $\mathsf{roll}$ is equivalent to
$\mathsf{map}$ping the $\mathsf{fold}$ over all the sub-expressions,
which means that $\mathsf{fold}$ interprets all of the constructors
homomorphically using the provided interpretation $\mathsf{f}$. Then
the $\eta$ rule says that $\mathsf{fold}$ is the \emph{unique} such
homomorphism, i.e. anything that satisfies the recurrence equation of
the $\mathsf{fold}$ is equal to it.

This definition as an initial algebra is well-understood semantically
but the $\eta$ principle in particular is somewhat cumbersome to use
directly in proofs. In dependent type theory, we would have a
dependent \emph{elimination} principle, which can be used to
implement functions by recursion as well as proofs by
induction. Unfortunately, since linear types do not support dependency on
linear types, we cannot directly adapt this approach. However, if we
are trying to prove that two morphisms out of a mutually recursive
type are equal, we can use the \emph{equalizer} type to prove their
equality by induction. That is, if our goal is to prove two functions
$f,g : \uparrow(\mu F\,x \lto A\,x)$ equal, it suffices to
implement a function $\mathsf{ind} : \uparrow(\mu F\,x \lto \equalizer
{a} f g)$ such that $\textrm{ind}(a) \equiv a$. Then an
inductive-style proof can be implemented by constructing
$\mathsf{ind}$ using a $\fold$. This can all be justified using only
the $\beta\eta$ principles for equalizers and inductive types, and
this is how our most complex inductive proofs are implemented in the Agda
formalization.
\pedro{I don't like the inductive proof explanation above, but I
  don't know how I would phrase it}

\subsection{Grammar-specific Additions}

So far, our calculus is a somewhat generic combination of dependent
types with non-commutative linear types. In order to carry out formal
grammar theory and define parsers, we need only add a few
grammar-specific constructions.

\theoryabbv is parameterized by a fixed, finite alphabet $\Sigma$ from which we
build our strings.  For each character $c \in \Sigma$,
we add a corresponding linear type $\literal{c}$. We can then define a
non-linear type $\CharGram$ as the disjunction of all of these characters,
and define a type $\StringGram$ as the Kleene star of $\CharGram$, i.e. as an
inductive linear type.
Then we add a function $\mathsf{read} :
\ltonl{(\top\lto \StringGram)}$ that intuitively ``reads'' the input
string from the input and makes it available. It is important that the
input type of $\mathsf{read}$ is $\top$, which can control any amount
of resources, and not $\I$ which controls no resources.

\begin{axiom}
  \label{ax:string-top}
   $\lambda s. \mathsf{read}(!(s)) \equiv \lambda
s. s$ where $!$ is the unique function $\ltonl(\StringGram\lto \top)$.
\end{axiom}

If we have a string, but then throw it away and read
it from the input, then we, in fact, recover the original string.
This ensures that the elements of the $\StringGram$
type always stand for the actual input string in our reasoning. In the
next section, we will show how these basic principles are enough to
provide a basis for verified parsing and formal grammar theory.

\section{Formal Grammar Theory in \theoryname}
\label{sec:applications}
This section explores the applications of \theoryabbv to formal
grammar theory. We demonstrate that several classical notions and constructions
integral to the theory of formal languages are faithfully represented
in \theoryname.

In the theory of formal grammars, there are two different notions of
equivalence: up to weak generative capacity, meaning just which strings are
accepted by the grammar; and up to \emph{strong} generative capacity, when
the parse trees of the two grammars are isomorphic
\cite{chom1963}. Using linear types as grammars, we can define both of
these notions of equivalence in \theoryabbv.

\begin{definition}
  \label{def:weakequiv}
  Grammars $\A$ and $\B$ are \emph{weakly equivalent} if there exist
  parse transformers $\f : \ltonl{(\A \lto \B)}$ and $\g : \ltonl {(\B
    \lto \A)}$. $\A$ is a \emph{retract} of $\B$ if they are
  weakly equivalent and $\lambda a. g(f(a)) \equiv \lambda a.a$. They
  are \emph{strongly equivalent} if further the other composition is
  the identity, i.e., $\lambda b. f(g(b)) \equiv \lambda b.b$.
\end{definition}

A formal grammar $\A$ is ambiguous if there are multiple parse trees
for the \emph{same} string $\w$.  For example, $\a \oplus \a$ is
ambiguous because there are two parses of $\stringquote{a}$,
constructed using $\inl$ and $\inr$. On the other hand, a formal
grammar is unambiguous when there is at most one parse tree for any
input string. We can capture this notion as a type in \theoryabbv as follows:
\begin{definition}
  \label{def:unambig}
  A grammar $\A$ is \emph{unambiguous} if for every linear type $\B$,
  $\f : \ltonl{(\B \lto \A)}$, and $\g : \ltonl{(\B \lto \A)}$
  then $\f \equiv \g$.
\end{definition}
\cref{def:unambig} can be read more intuitively as stating that $\A$
is unambiguous if there is at most one way to transform parses of any
other grammar $\B$ into parses of $\A$. This notion of an unambiguous
type is the analog for linear types of the definition of a (homotopy)
\emph{proposition} in the terminology of homotopy type
theory \cite{hottbook}. The most basic unambiguous types are $\top$ and
$0$, and in a system of classical logic all unambiguous types would
have to be equivalent to one of these, but with our axioms we can show
also that $\I$ and literals $\literal c$ are unambiguous. To see this, first, we
establish two useful properties of unambiguity.
\begin{lemma}[\Agda]
  \label{lem:retract}
  If $\B$ is unambiguous and $\A$ is a retract of $\B$ then $\A$ is unambiguous.
\end{lemma}
\begin{lemma}[\Agda]
  \label{lem:unambig-sum}
  As a consequence of \cref{cor:binary-mono}, if a binary disjunction $A \oplus B$ is unambiguous then
  $A$ and $B$ are each unambiguous.
\end{lemma}
From \cref{lem:retract}, we can prove that $\StringGram$ is unambiguous,
since it is a retract of $\top$. In fact, observe that if $\A$ is a
retract of $\B$ and $\B$ is unambiguous, then in fact $\A$ and $\B$
are strongly equivalent, as the equation $\lambda b. f(g(b)) \equiv
\lambda b. b$ follows because $\B$ is unambiguous. Therefore
$\StringGram$ is also strongly equivalent to $\top$. Next, since $\StringGram$ is
defined as a Kleene star, we can easily show that $\StringGram \cong \I
\oplus \CharGram \oplus (\CharGram \otimes \CharGram \otimes \StringGram)$. Then by
\cref{lem:unambig-sum}, we have that $\I$ and
$\CharGram$ are unambiguous as well. Using the finiteness of the alphabet
$\Sigma$ and the unambiguity of $\CharGram$,
we have that each literal $\literal c$ is likewise unambiguous.

We now turn to our main task, which is using our linear type system to
implement verified parsers. Given a grammar defined as a linear type
$A$, a first attempt at defining a parser would be to implement a
function $\uparrow{(\StringGram \lto A)}$. But since our linear functions
must be total, this means that we can construct an $A$ parse for
\emph{every} input string, which is impossible for most grammars of
interest. Instead we might try to write a partial function as a
$\ltonl{(\StringGram \lto (A \oplus \top))}$ using the ``option''
monad. This allows for the possibility that the input string does not
parse, but is far too weak as a specification: we can trivially
implement a parser for any type by always returning $\inr$. The
correct notion of a parser should be one that allows for failure, but
only in the case that a parse cannot be constructed.

\begin{definition}
  \label{def:disjoint}
  Linear types $A$ and $B$ are \emph{disjoint} if there is a function
  $\ltonl {\left( A \& B \lto 0 \right)}$.
\end{definition}

\begin{definition}
  \label{def:parser}
  A parser for a linear type $A$ is the choice a type $A_{\neg}$ disjoint from
  $A$ and function $\uparrow{(\StringGram \lto
    A \oplus A_{\neg})}$.
\end{definition}
Here we replace $\top$ in our partial parser type with a type
$A_{\neg}$ that we can think of as a negation of $A$. The function
$\uparrow{(A \& A_{\neg} \lto 0)}$ ensures that it is impossible for
$A$ and $A_{\neg}$ to parse the same input string. This means that in
defining a parser, we will need to define a kind of negative grammar
for strings that do not parse.

Fortunately, we will see that
deterministic automata naturally support such a notion with no
additional effort: the negative grammar is simply the grammar for
traces that end in a rejecting state. This follows from the following
principle, a consequence of \cref{ax:disjointness}.

\begin{lemma}[\Agda]
  \label{lem:unambig-to-disjoint}
  If $\LinSigTy{x}{X}{A\,x}$ is unambiguous, then for $x \neq x'$, $A\,x$ and $A\,x'$ are
  disjoint. In particular, if the binary product $A \oplus A_{\neg}$ is unambiguous, then
  $A$ and $A_{\neg}$ are disjoint.
\end{lemma}

\begin{proof}
If $\LinSigTy{x}{X}{A\,x}$ is unambiguous, then all functions into it from
$A\,x \& A\,x'$ are equal. In particular,
$\sigma\,x\,\circ\,\pi_{1} \equiv \sigma\,x'\,\circ\,\pi_{2}$ so there is a
function \(
\ltonl{\left(A\,x \& A\,x' \lto
  \equalizer{b}{\left( \sigma\,x \circ \pi_1 \right)}{\left(  \sigma\,x'\circ \pi_2 \right)}
\right)}.
\)
We then compose with the function in \cref{ax:disjointness} to prove that $A\,x$ and $A\,x'$ are disjoint.
\end{proof}

Writing a parser as a linear term intrinsically verifies the \emph{soundness} of
the parser for free from the typing: any $\inl$ parse that we return \emph{must}
correspond to a parse tree of the input string. Further, if we verify the
disjointness property of \cref{def:parser} we then also get the
\emph{completeness} of the parser as well, that when the parser rejects the
input that there are no valid parses.

Our main method for constructing verified parsers is to show that a
grammar $A$ is weakly equivalent to a grammar for a deterministic
automaton. Parsers for deterministic automata are simple to implement
by stepping through the states of the automaton, with the rejecting
traces serving as the negative grammar. This is sufficient due to the
following:
\begin{lemma}[\Agda]
  \label{lem:wk-eqv-parse}
  If $\A$ is weakly equivalent to $\B$ then any parser for $\A$ extends to a
  parser for $\B$.
\end{lemma}
Here we need both directions of the weak equivalence. We need $A \lto
B$ to extend the parser from $\StringGram \lto A \oplus A_{\neg}$ to
$\StringGram \lto B \oplus A_{\neg}$, but then we also need $B \lto A$ to
establish that $A_{\neg}$ is disjoint from $B$.

\subsection{Regular Expressions and Finite Automata}

Next, we describe how to construct an intrinsically
verified parser for regular expressions by compiling it to an NFA and
then a DFA. That is, for each regular expression $A$, we construct an
NFA $N(A)$ and a corresponding DFA $D(A)$ such that $A$ is strongly
equivalent to the traces of $N(A)$ and weakly equivalent to the
accepting traces of $D(A)$. Then we can easily construct a parser for
traces of $D(A)$ and apply \cref{lem:wk-eqv-parse} to get
a verified regular expression parser.

A regular expression in \theoryabbv is a linear type constructed
using only the connectives $\literal c$, $0$, $\oplus$, $I$,
$\otimes$, and Kleene star.  In \cref{sec:type-theory-examples}, we
saw one particular NFA and its corresponding type of traces. More
generally, in \cref{fig:nfatrace} we define a linear type of traces through an
arbitrary NFA $\N$.

$\Trace_{\N}$ is an inductive type indexed by the starting
state of the trace $\s : \N.\states$, and it may be built
through one of three constructors. We may terminate a trace at an
accepting state with the constructor $\nil$. Here we use an Agda-style
Unicode syntax for $\bigamp$, as well as using the function arrow to
mean a non-dependent version of $\bigamp$. If we had a trace beginning
at the destination state of a transition, then we may use the $\cons$
constructor to combine that trace with a parse of the label
of the transition to build a trace beginning at the source of the
transition.  Finally, if we had a trace beginning at the destination
of an $\epsilon$-transition then we may use $\epscons$ to pull it back
along the $\epsilon$-transition and construct a trace beginning at the
source of the $\epsilon$-transition. As a shorthand, we write $Parse_{N}$
for the accepting traces out of $N.init$.
\newcommand{\D}{\texttt{D}}
\begin{figure}
\begin{floatlisting}
data Trace(*@$_{\color{black}\N}$@*) : (s : N .states) (*@$\tocode$@*) L where
  nil : (*@$\uparrowcode$@*)(&[ s : N .states ] N .isAcc s (*@$\tocodelin$@*) Trace(*@$_{\color{black}\N}$@*) s)
  cons : (*@$\uparrowcode$@*)(&[ t : N .transitions ] ('#N .label t#' (*@$\lto$@*) Trace(*@$_{\color{black}\N}$@*) (N .dst t)
                                                 (*@$\lto$@*) Trace(*@$_{\color{black}\N}$@*) (N .src t)))
  (*@$\color{linconstructorcolor}\epsilon$@*)cons : (*@$\uparrowcode$@*)(&[ t : N .(*@$\color{nonlinconstructorcolor}\epsilon$@*)transitions ] (Trace(*@$_{\color{black}\N}$@*) (N .(*@$\color{nonlinconstructorcolor}\epsilon$@*)dst t) (*@$\lto$@*) Trace(*@$_{\color{black}\N}$@*) (N .(*@$\color{nonlinconstructorcolor}\epsilon$@*)src t)))
data Trace(*@$_{\color{black}\D}$@*) : (s : (*@\color{black}D@*) .states) (b : Bool) (*@$\tocode$@*) L where
  nil : (*@$\uparrowcode$@*)(&[ s : (*@\color{black}D@*) .states ] Trace(*@$_{\color{black}\D}$@*) s (*@\color{black}D@*) .isAcc s)
  cons : (*@$\uparrowcode$@*)(&[ c : (*@$\color{nonlintycolor}\Sigma$@*) ] &[ s : (*@\color{black}D@*) .states ]
           &[ b : Bool ] ('(*@\color{alphabetcolor}c@*)' (*@$\lto$@*) Trace(*@$_{\color{black}\D}$@*) ((*@\color{black}D@*) .(*@$\color{nonlinconstructorcolor}\delta$@*) c s) b (*@$\lto$@*) Trace(*@$_{\color{black}\D}$@*) s b))
\end{floatlisting}
\caption{Traces of an NFA $N$ and a DFA $D$}
\label{fig:nfatrace}
\end{figure}

$\Trace_{\D}$, the linear type of traces through $\D$, is given
next. Unlike traces for an NFA, we parameterize this type additionally
by a boolean which says whether the trace is accepting or
rejecting. These traces may be terminated in an accepting state $\s$
with the $\nil$ constructor.  The $\cons$ constructor builds a trace
out of state $\s$ by linearly combining a parse of some character $\c$
with a trace out of the state $\D.\delta~\c~\s$. The trace built with
$\cons$ is accepting if and only if the trace out of $D.\delta~c~s$ is
accepting.

Because DFAs are deterministic, we are able to prove that their types
of traces are unambiguous and define a parser for them directly. In particular
we show that for any start state $s$, ${\LinSigTy {b} {\Bool}
  \Trace_{D}~s~b}$ is a retract of $\StringGram$ and apply \cref{lem:retract} to
derive unambiguity. That is, we first
construct a function $\pparse_{D}$ which is a parser for
$\Trace_D~{s}~{true}$ with $\Trace_D~{s}~{false}$ being the disjoint
type used, and disjointness follows from the
unambiguity of $\LinSigTy {b} {\Bool} {\Trace_{D}~s~b}$ by \cref{lem:unambig-sum}.

The parser,
$\pparse_{D}$,
is defined by recursion on strings in \cref{fig:printdfa}. If this string is empty, then $\pparse_{D}$
terminates a trace at the input state $s$. If the
string is nonempty, then $\pparse_{D}$ walks forward in $D$ from the input state $s$ by
the character at the head of the string.
The inverse, $\print_{D}$ is defined by recursion on traces. If the trace is defined
via $\nil$, then $\print_{D}$ returns the empty string. Otherwise, if the trace
is defined by $\cons$ then $\print_{D}$ appends the character from the most
recent transition to the output
string and recurses.
We prove this is a retraction by induction on traces.

\begin{theorem}[\Agda]
  \label{thm:dfa-parser}
  $\pparse_D s$ is a parser for $\Trace_D~s~\true$.
\end{theorem}

\begin{figure}
\begin{lstlisting}
parse(*@$\color{black}_D$@*) : (*@$\uparrowcode$@*)(String (*@$\lto$@*) &[ s : (*@\color{black}D@*) .states ] (*@$\oplus$@*)[ b : Bool ] Trace(*@$_{\color{black}\D}$@*) s b)
parse(*@$\color{black}_D$@*) String .nil s = (*@$\color{linconstructorcolor}\sigma$@*) ((*@\color{black}D@*) .isAcc s) (Trace(*@$_{\color{black}\D}$@*) .nil s)
parse(*@$\color{black}_D$@*) (String .cons ((*@$\color{linconstructorcolor}\sigma$@*) c a) w) s = let (*@$\color{linconstructorcolor}\sigma$@*) b t = parse w ((*@\color{black}D@*) .(*@$\color{nonlinconstructorcolor}\delta$@*) c s) in
                                    (*@$\color{linconstructorcolor}\sigma$@*) b (Trace(*@$_{\color{black}\D}$@*) .cons c s b a t)
print(*@$\color{black}_D$@*) : (s : (*@\color{black}D@*) .states) (*@$\tocode$@*) (*@$\uparrowcode$@*)(((*@$\oplus$@*)[ b : Bool ] Trace(*@$_{\color{black}\D}$@*) s b) (*@$\lto$@*) String)
print(*@$\color{black}_D$@*) s ((*@$\color{linconstructorcolor}\sigma$@*) b (Trace(*@$_{\color{black}\D}$@*) .nil .s)) = String .nil
print(*@$\color{black}_D$@*) s ((*@$\color{linconstructorcolor}\sigma$@*) b (Trace(*@$_{\color{black}\D}$@*) .cons c ((*@\color{black}D@*) .(*@$\color{nonlinconstructorcolor}\delta$@*) c .s) b a trace)) =
  String .cons ((*@$\color{linconstructorcolor}\sigma$@*) c a) (print(*@$\color{black}_D$@*) ((*@\color{black}D@*) .(*@$\color{nonlinconstructorcolor}\delta$@*) c s) ((*@$\color{linconstructorcolor}\sigma$@*) b trace))
\end{lstlisting}
\caption{Parser/printer for DFA traces}
\label{fig:printdfa}
\end{figure}

Working backwards, we can then show the traces of an NFA are weakly
equivalent to the traces of a DFA implementing a variant of Rabin and
Scott's classic powerset construction
\cite{rabinFiniteAutomataTheir1959}. Here we note that this is
\emph{only} a weak equivalence and not a strong equivalence, as the
DFA is unambiguous even if the NFA is not.
\begin{construction}[Determinization, \Agda]
  \label{cons:determinization}
  Given an NFA $N$, we construct a DFA $D$ such that
  $Parse_{N}$ is weakly equivalent to $Parse_{D}$.
\end{construction}
\newcommand{\X}{\texttt{X}}
\newcommand{\Y}{\texttt{Y}}
\newcommand{\x}{\texttt{x}}
\newcommand{\y}{\texttt{y}}
\begin{proof}
  Define the states of $\D$ to be $\mathbb{P}_{\epsilon}(\N.\states)$ --- the type
  of $\epsilon$-closed\footnote{A subset of states $\X$ is $\epsilon$-closed if
    for every $\s \in \X$ and $\epsilon$-transition $\s \overset{\epsilon}{\to} \s'$ we have $\s' \in \X$.} subsets of
  $\N.\states$. A subset is accepting in $\D$ if it contains an accepting state
  from $\N$. The initial state of $\D$ is the
  $\epsilon$-closure of $\N.\init$. Lastly, the transition function
  of $\D$ sends the subset $X$ under the character $\c$ to the
  $\epsilon$-closure of all the states reachable from $X$ via a transition
  labeled with the character $\c$.

  We demonstrate the weak equivalence between $Parse_{N}$ and
  $Parse_{D}$ by constructing parse transformers between the two
  grammars. To build the parse transformer
  $\ltonl{(Parse_{N} \lto Parse_{D})}$, we strengthen our inductive hypothesis to quantify over every start state
  and build a term
\newcommand{\NtoD}{\texttt{NtoD}}
\newcommand{\DtoN}{\texttt{DtoN}}
  \(
    \NtoD : \ltonl {\left(\Trace_{\N}~\s~\true \lto \LinPiTy{\X}{\D.\states}{\LinPiTy{\texttt{sInX}}{\X \ni \s}{ \Trace_{\D}~\X~\true}} \right) }
  \)
  that maps a trace in $\N$ from an arbitrary state $\s$ to a trace in $\D$
  that may begin at any subset of states $\X$ that contains $\s$. $\NtoD$ may then be instantiated at
  $\s = \N.\init$ and $\X = D.\mathsf{init}$ to get the desired
  parse transformer.

To construct a term from DFA traces to NFA traces, we similarly strengthen our
inductive hypothesis and build a parse transformer that ranges over arbitrary $\epsilon$-closed subsets $X : \mathbb{P}_{\epsilon}{\N.\states}$,
\(
  \DtoN : \ltonl {\left( \Trace_{\D}~\X~\true \lto \LinSigTy{s}{\N.\states}{\LinSigTy{\texttt{sInX}}{\X \ni \s}{ \Trace_{\N}~s~\true}} \right)}
\).
Because the states of $D$ are $\epsilon$-closed subsets of $N.\states$, if two
states $s$ and $s'$ belong to the the same $\epsilon$-closed subset
$X : \mathbb{P}_{\epsilon}(N.\states)$ then there exists a $\epsilon$-path in
$N$ between the two, but we do not necessarily know which one. Similarly, the
data contained in the type $\Trace_{\D}~X~true$ is the \emph{existence} of an accepting trace in
$\N$ beginning at some state in $\D$.

To define $\DtoN$, we need a choice function that extracts out a trace in $N$
from the mere existence of one. We achieve this by choosing the smallest trace
through $\N$ subject to an ordering on the non-linear types $N.\states$,
$N.\transitions$, and $N.\epstransitions$. In essence, the given orderings
specify a global disambiguation strategy.
\end{proof}

Finally, given any regular expression we can construct a
\emph{strongly} equivalent NFA. While only weak equivalence is
required to construct a parser, proving the strong equivalence shows
that other aspects of formal grammar theory are also verifiable in
\theoryabbv.
\newcommand{\R}{\texttt{R}}
\begin{construction}[Thompson's Construction \cite{thompsonProgrammingTechniquesRegular1968}, \Agda]
  \label{cons:thompson}
  Given a regular expression $\R$, we build an NFA $\N$ such that $\R$ is
  strongly equivalent to $\Trace_{N}(\N.\init)$.
\end{construction}

\begin{corollary}[\Agda]
  We may build a parser for every regular expression $\R$.
\end{corollary}
\begin{proof}
We combine the strong equivalence of \cref{cons:thompson} with the weak
equivalence of \cref{cons:determinization} to show that $\R$ is weakly
equivalent to the traces of its determinized automaton. Then, we use
\cref{lem:wk-eqv-parse} to extend the parser from \cref{thm:dfa-parser}
with respect to this weak equivalence.
\end{proof}

\subsection{Context-free grammars}

Next, we give two examples for parsing context-free
grammars (CFGs). CFGs can be encoded in our type
theory similarly to regular expressions, as CFGs are equivalent
to the formalism of \emph{$\mu$}-regular expressions, where the Kleene
star is replaced by an arbitrary fixed point
operation \cite{leis_towards_1992}.

A simple example of a CFG is the Dyck grammar of balanced
parentheses, which we define in \Cref{fig:dyckinductive}.
$\Dyck$ is a grammar over the alphabet $\{ \stringquote{(}, \stringquote{)} \}$. The $\nil$ constructor shows that
the empty string is balanced, and the $bal$ constructor builds a balanced
parse by wrapping an already balanced parse in an additional set of parentheses
then following it with another balanced parse. We construct a parser for $\Dyck$ by building a deterministic automaton $M$ such
that $Parse_{M}$ is strongly equivalent to $\Dyck$.
\begin{figure}
\begin{floatlisting}
data Dyck : L where
  nil : (*@$\uparrowcode$@*) Dyck
  bal : (*@$\uparrowcode$@*)('(' (*@$\lto$@*) Dyck (*@$\lto$@*) ')' (*@$\lto$@*) Dyck (*@$\lto$@*) Dyck)
\end{floatlisting}
\caption{The Dyck grammar as an inductive linear type}
\label{fig:dyckinductive}
\end{figure}

We implement the parser for the Dyck language using an \emph{infinite} state deterministic
automaton, in \Cref{fig:DyckAutomaton}. Here the state is a ``stack''
counting how many open parentheses have been seen so far. Functions
$parse_{M}$ and $print_{M}$ for this automaton can be defined
analogously to those for DFAs, and so $\LinSigTy {s}
{M .states} {\LinSigTy {b}{Bool} {Trace_{M}~s~b}}$ is likewise
unambiguous.

\begin{theorem}[\Agda]
  \label{thm:dyck}
  $Dyck$ and $Parse_{M}$ are strongly equivalent, so we may build a parser for $\Dyck$.
\end{theorem}

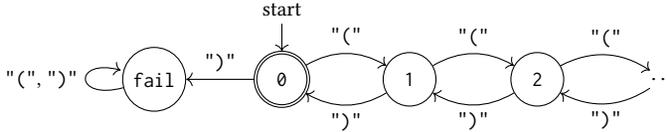
\begin{figure}
  \footnotesize
  \begin{tikzpicture}[node distance = 17mm ]
    \node[state, initial above, accepting] (0) {$0$};
    \node[state, right of=0] (1) {$1$};
    \node[state, right of=1] (2) {$2$};
    \node[right of=2] (3) {$\dots$};
    \node[state, left of=0] (fail) {$fail$};

    \path[->] (0) edge[above, bend left] node{\stringquote{(}} (1)
              (1) edge[below, bend left] node{\stringquote{)}} (0)
              (1) edge[above, bend left] node{\stringquote{(}} (2)
              (2) edge[below, bend left] node{\stringquote{)}} (1)
              (2) edge[above, bend left] node{\stringquote{(}} (3)
              (3) edge[below, bend left] node{\stringquote{)}} (2)
              (fail) edge[loop left] node{\stringquote{(}, \stringquote{)}} (fail)
              (0) edge[above] node{\stringquote{)}} (fail);
  \end{tikzpicture}
  \caption{Automaton $M$ for the Dyck grammar}
  \label{fig:DyckAutomaton}
\end{figure}

Our final example is of a simple grammar of arithmetic expressions
with an associative operation. Here we take the alphabet to be $\{
\stringquote{(}, \stringquote{)}, \stringquote{+} , \stringquote{NUM}
\}$. In \Cref{fig:binop-inductive} we define the expression grammar
using two mutually
recursive types, corresponding to the two non-terminals we would use
in a CFG syntax. The syntactic structure encodes that the binary
operation is right associative. In the same figure, we define the
traces of an automaton with one token of lookahead. The automaton\footnote{In
  the automaton definition,
  $\texttt{NotStartsWithLP}$ is defined as
  \(I \oplus
  (\literal{)} \oplus \literal{+} \oplus \literal{NUM}) \otimes \top
  \). Similarly,
  $\texttt{NotStartsWithRP}$ is defined as
  \(I \oplus
  (\literal{(} \oplus \literal{+} \oplus \literal{NUM}) \otimes \top
  \).
}
has four different ``states'', each with access to a natural number
``stack''. The ``opening'' state $O$ expects either a left paren, in
which case it increments the stack and stays in the opening state, or
sees a number and proceeds to the $D$ state. The ``done opening''
state $D$ is where lookahead is used: if the next token will be a
right paren, then we proceed to $C$; otherwise, we proceed to $A$.
In the ``closing'' state $C$ if
we observe a right paren, then we decrement the count and continue to the
$D$ state. In the ``adding'' state $A$, we succeed if the string
ends and the count is $0$; otherwise, if we see a plus we continue to the $O$
state. Additionally, since the automaton need parse all of the
incorrect strings, we add all of the failing cases.

\begin{theorem}[\Agda]
  \label{thm:exp-parser}
  We construct a parser for $\mathsf{Exp}$ by showing it is weakly
  equivalent to $\mathsf{O}~0~\true$.
\end{theorem}

With \cref{ax:dist}, it is straightforward to implement a parser for this lookahead automaton,
generalizing the approach for deterministic automata. Without access to
distributivity, we may define a rudimentary lookahead operation via the chain of
equivalences,
\[
  A~n~b \cong (A~n~b) \& \top \cong (A~n~b)\,\& \left(\left(\literal{)} \otimes \top \right) \oplus \texttt{NotStartsWithRP}\right)
\]
When defining a parser, if the lookahead character is a right paren we would
like to apply $\texttt{lookAheadRP}$; otherwise, apply $\texttt{lookAheadNot}$.
However, without distributivity we have no means of relating the bit of
information learned by looking ahead to the input $A~n~b$ parse.
That is, to choose which constructor to apply we would need the end of this
chain of equivalences to be a binary sum rather than a binary product,
necessitating the distributivity axiom.

\begin{figure}
\begin{floatlisting}
data Exp : L where
  done : (*@$\uparrowcode$@*)(Atom (*@$\lto$@*) Exp)
  add : (*@$\uparrowcode$@*)(Atom (*@$\lto$@*) '+' (*@$\lto$@*) Exp (*@$\lto$@*) Exp)
data Atom : L where
  num : (*@$\uparrowcode$@*)('#NUM#' (*@$\lto$@*) Atom)
  parens : (*@$\uparrowcode$@*)('(' (*@$\lto$@*) Exp (*@$\lto$@*) ')' (*@$\lto$@*) Atom)

data O : Nat (*@$\tocode$@*) Bool (*@$\tocode$@*) L where
  left : (*@$\uparrowcode$@*)(&[ n : Nat ] &[ b : Bool ] '(' (*@$\lto$@*) O (n + 1) b (*@$\lto$@*) O n b)
  num : (*@$\uparrowcode$@*)(&[ n : Nat ] &[ b : Bool ] '#NUM#' (*@$\lto$@*) D n b (*@$\lto$@*) O n b)
  unexpected : (*@$\uparrowcode$@*)(&[ n : Nat ] (NotStartsWithLP (*@$\lto$@*) O n false)
data D : Nat (*@$\tocode$@*) Bool (*@$\tocode$@*) L where
  lookAheadRP : (*@$\uparrowcode$@*)(&[ n : Nat ] &[ b : Bool ] ((')' (*@$\otimes$@*) (*@$\color{lintycolor}\top$@*)) (*@$\&$@*) C n b) (*@$\lto$@*) D n b)
  lookAheadNot : (*@$\uparrowcode$@*)(&[ n : Nat ] &[ b : Bool ] (NotStartsWithRP (*@$\&$@*) A n b) (*@$\lto$@*) D n b)
data C : Nat (*@$\tocode$@*) Bool (*@$\tocode$@*) L where
  closeGood : (*@$\uparrowcode$@*)(&[ n : Nat ] &[ b : Bool ] ')' (*@$\lto$@*) D n b (*@$\lto$@*) C (n + 1) b)
  closeBad : (*@$\uparrowcode$@*)(')' (*@$\lto$@*) C 0 false)
  unexpected : (*@$\uparrowcode$@*)(&[ n : Nat ] NotStartsWithRP (*@$\lto$@*) C n false)
data A : Nat (*@$\tocode$@*) Bool (*@$\tocode$@*) L where
  doneGood : (*@$\uparrowcode$@*)(A 0 true)
  doneBad : (*@$\uparrowcode$@*)(&[ n : Nat ] A (n + 1) false)
  add : (*@$\uparrowcode$@*)(&[ n : Nat ] &[ b : Bool ] '+' (*@$\lto$@*) O n b (*@$\lto$@*) A n b)
  unexpected : (*@$\uparrowcode$@*)(&[ n : Nat ] ('(' (*@$\oplus$@*) ')' (*@$\oplus$@*) '#NUM#') (*@$\lto$@*) (*@$\color{lintycolor} \top$@*) (*@$\lto$@*) A n false)
\end{floatlisting}
\caption{Associative arithmetic expressions and a corresponding lookahead automaton}
\label{fig:binop-inductive}
\end{figure}

\subsection{Unrestricted Grammars}

While we have shown only examples for context-free grammars, in fact
arbitrarily complex grammars are encodable in \theoryabbv. To
demonstrate this, we show that for any \emph{non-linear} function $P :
\StringSem \to U$, where here $\StringSem$ is the \emph{non-linear} type
of strings over the alphabet, we can construct a grammar whose parses
correspond to $P$.
\(
Reify~P = \LinSigTy{w} {\StringSem} {\LinSigTy {x} {P~w} {~\lceil w \rceil}}
\)
where $\lceil \stringquote{} \rceil = I$ and
$\lceil c :: w \rceil = \literal c \otimes \lceil w \rceil$.

This reification operation on functions $\StringSem \to U$ is very expressive,
as it sidesteps our linear typing connectives and utilizes the whole of nonlinear
dependent type theory to define a grammar. For example, given a
Turing machine $T$ one may
define a non-linear predicate $accepts : \StringSem \to U$ that encodes that $T$
halts and accepts an input string. Then, $Reify~accepts$ is a linear type that
captures precisely the string the same language as $T$. In general,
$Lang(Reify~accepts)$ is recursively enumerable --- the most general class of
languages in the Chomsky hierarchy.
\begin{construction}[\Agda]
  \label{cons:turing}
  For any Turing machine $T$, we can construct a grammar in \theoryabbv
  that accepts the same language as $T$.
\end{construction}

\section{Denotational Semantics and Implementation}
\label{sec:semantics-and-metatheory}

To justify our assertion that \theoryabbv is a syntax for formal
grammars and parse transformers, we will now define a
\emph{denotational semantics} that makes this mathematically precise
by defining a notion of formal grammar and parse transformer then
showing that our type theory can be soundly interpreted in this
model. We then discuss how this denotational semantics provides the
basis for our prototype implementation in Agda.

\subsection{Formal Grammars and Parse Transformers}

The most common definition of a formal grammar is as generative
grammars, defined by a set of non-terminals, a specified start symbol
and set of production rules. We instead use a more abstract
formulation that is closer in spirit to the standard definition of a
formal \emph{language} \cite{elliottSymbolicAutomaticDifferentiation2021}:
\begin{definition}
  A \emph{formal language} $L$ is a function from strings to propositions.
  A (small) \emph{formal grammar} $A$ is a function from strings to (small) sets.
\end{definition}
We think of the grammar $A$ as taking a string to the set of all parse
trees for that string. However since $A$ could be any function
whatsoever there is no requirement that an element of $A(w)$ be a
``tree'' in the usual sense. This definition provides a simple,
syntax-independent definition of a grammar that can be used for any
formalism: generative grammars, categorial grammars, or our own
type-theoretic grammars. Note that the definition of a formal grammar
is a generalization of the usual notion of formal language since a
proposition can be equivalently defined as a subset of a one-element
set. Then the difference between a formal grammar and a formal
language is that formal grammars can be \emph{ambiguous} in that there
can be more than one parse of the same string. Even for unambiguous
grammars, we care not just about \emph{whether} a string has a parse
tree, but \emph{which} parse tree it has, i.e., what the structure of
the element of $A(w)$ is.  To interpret our universes $U , L$ we
assume we have a universe of \emph{small} sets. In the remainder, all
formal grammars are assumed to be small.

We then interpret linear \emph{terms} as \emph{parse transformers}:
\begin{definition}
  Let $A_1,A_2$ be formal grammars. Then a parse transformer $f$ from
  $A_1$ to $A_2$ is a function assigning to each string $w$ a function
  $f_w : A_1(w) \to A_2(w)$.
\end{definition}
Just as formal grammars generalize formal languages, parse
transformers generalize formal language inclusion: if $A_1(w), A_2(w)$
are all subsets of a one-element set, then a parse transformer is
equivalent to showing that $A_1(w) \subseteq A_2(w)$. In our
denotational semantics, linear terms will be interpreted as such parse
transformers, and the notions of unambiguous grammar, parsers,
disjointness, etc, introduced in \Cref{sec:applications} can be
verified to correspond to their intended meanings under this
interpretation.

Parse transformers can be composed: given two parse transformers $f$ and $g$,
their composition is defined pointwise, i.e. $(f\circ g)_w = f_w \circ g_w$.
Furthermore, given a formal grammar $A$, its identity transformer is $id_w =
id_{A(w)}$, where $id_{A(w)}$ is the identity function on the set $A(w)$. This
defines a \emph{category}.
\begin{definition}
  Define $\Grammar$ to be the category whose objects are formal
  grammars and morphisms are parse transformers.
\end{definition}
This category is equivalent to the slice category $\Set/\Sigma^*$ and
is very well-behaved. It is complete, co-complete, Cartesian
closed and carries a monoidal biclosed structure. We will use these
structures to model the linear types, terms and equalities in
\theoryabbv.

\subsection{Semantics}

We now define our denotational semantics.
\begin{definition}[Grammar Semantics]
  We define the following interpretations by mutual recursion on the
  judgments of \theoryabbv:
  \begin{enumerate}
  \item For each non-linear context $\Gamma \isCtx$, we define a set $\sem \Gamma$.
  \item For each non-linear type $\Gamma \vdash X \isTy$, and element
    $\gamma \in \sem\Gamma$, we define a set $\sem X \gamma$.
  \item For each linear type $\Gamma \vdash A \isLinTy$ and element $\gamma
    \in \sem\Gamma$, we define a formal grammar $\sem{A}\gamma$. We
    similarly define a formal grammar $\sem\Delta\gamma$ for each
    linear context $\linctxwff \Gamma \Delta$.
  \item For each non-linear term $\Gamma \vdash M : X$ and $\gamma \in \sem{\Gamma}$, we define an element $\sem{M}\gamma \in \sem{X}\gamma$.
  \item For each linear term $\Gamma; \Delta \vdash e : A$ and $\gamma \in \sem{\Gamma}$, we define a parse transformer from $\sem{\Delta}\gamma$ to $\sem{A}\gamma$.
  \end{enumerate}
  And we verify the following conditions:
  \begin{enumerate}
  \item If $\Gamma \vdash X \isSmall$, then $\sem X \gamma$ is a small set.
  \item If $\Gamma \vdash X \equiv X'$ then for every $\gamma$, $\sem{X}\gamma = \sem{X'}\gamma$.
  \item If $\Gamma \vdash A \equiv A'$ then for every $\gamma$, $\sem{A}\gamma = \sem{A'}\gamma$.
  \item If $\Gamma \vdash M \equiv M' : X$ then for every $\gamma$, $\sem{M}\gamma = \sem{M'}\gamma$.
  \item If $\Gamma;\Delta \vdash e \equiv e' : A$ then for every $\gamma$, $\sem{e}\gamma = \sem{e'}\gamma$.
  \end{enumerate}
\end{definition}

The interpretation of dependent types as sets is standard
\cite{Hofmann_1997}. We present the concrete descriptions of the
semantics of linear types, as well as our non-standard non-linear
types in \Cref{fig:semantics}. The grammar for a literal $c$ has a
single parse precisely when the input string consists of the single
character. The grammar for the unit similarly has a single parse for
the empty string.  A parse of the tensor product $A \otimes B$
consists of a \emph{splitting} of the empty string into a prefix $w_1$
and suffix $w_2$ along with an $A$ parse of $w_1$ and $B$ parse of
$w_2$. A parse of $\bigoplus_{x:X} A$ is a pair of an element of the
set $X$ and a parse of $A(x)$, while dually a parse of $\bigwith_{x:X}
A$ is a \emph{function} taking any $x:X$ to a parse of $A(x)$. A
$w$-parse of $A \lto B$ is a function that takes an $A$ parse of some
other string $w'$ to a $B$ parse of $ww'$, and $B \tol A$ is the same
except the $B$ parse is for the reversed concatenation $w'w$.
The set $\uparrow A$ is the set of parse for the empty string for
$A$. This definition means that $\sem{\uparrow (A \lto B)}$ (or $\sem{\uparrow (B \tol A)}$) is
equivalent to the set of parse transformers:
\( \sem{\uparrow (A \lto B)}\gamma = \sem{A \lto B}\gamma\,\epsilon = \prod_{w'} \sem{A}\gamma\,w' \to \sem{B}\gamma\,w'\).

Next, a parse in the equalizer $\equalizer{a}{f}{g}$ is defined as
a parse in $\sem{A}$ that is mapped to the same parse by the parse transformers $\sem{f}$ and
$\sem{g}$.
The universe $L$ of linear types is interpreted as the set of all
small grammars.
\begin{figure}
  \begin{minipage}[t]{.3\textwidth}
  \begin{footnotesize}
    \begin{flalign*}
    &\sem{c}\gamma\,w = \{ c | w = c\}\\
    &\sem{I}\gamma\,w = \{ () \pipe w = \epsilon \}\\
    &\sem{\uparrow A}\gamma = \sem{A}\gamma\,\epsilon\\
    &\sem{A \lto B}\gamma\,w = \prod_{w'} \sem{A}\gamma\,w' \to \sem{B}\gamma\,ww'\\
    &\sem{B \tol A}\gamma\,w = \prod_{w'} \sem{A}\gamma\,w' \to \sem{B}\gamma\,w'w\\
    &\sem{\el(F)}\gamma\,G = \sem{F}\gamma\,G\\
    &\sem{\map(F)}\gamma\,f = \sem{F}\gamma\,f\\
    &\sem{\mu A}\gamma = \mu (\sem{A}\gamma)
  \end{flalign*}
  \end{footnotesize}
  \end{minipage}%
  \begin{minipage}[t]{.7\textwidth}
  \begin{footnotesize}
    \begin{flalign*}
    &\sem{A \otimes B}\gamma\,w = \{ (w_1,w_2,a,b) \pipe w_1w_2 = w \wedge a \in \sem{A}\gamma\,w_1 \wedge b \in \sem{B}\gamma\,w_2 \}\\
    &\sem{\bigoplus_{x:X} A}\gamma\,w = \coprod_{x \in \sem{X}\gamma} \sem{A}(\gamma, x)\,w \\
    &\sem{\bigwith_{x:X} A}\gamma\,w = \prod_{x\in \sem{X}\gamma} \sem{A}(\gamma, x)\,w \\
    &\sem{\equalizer{a}{f}{g}}\gamma\,w = \{ a \in \sem{A}\gamma\,w \pipe \sem{f}\gamma\,w\,a = \sem{g}\gamma\,w\,a \} \\
    &\sem{L}\gamma = \Grammar_0\\
    &\sem{SPF\,X}\gamma = \textrm{DepPolyFunctor}(\sem{X}\gamma\times \Sigma^*,\Sigma^*)
  \end{flalign*}
  \end{footnotesize}
  \end{minipage}
  \caption{Grammar Semantics}
  \label{fig:semantics}
\end{figure}
The most complex part of the semantics is the interpretation of
strictly positive functors and indexed inductive linear types.  We
interpret a strictly positive functor as a \emph{dependent polynomial
functor} on the category of sets, also sometimes called an
\emph{indexed container} \cite{gambino_wellfounded_2004,altenkirch_indexed_2015}.
\begin{definition}
  Let $I$ and $O$ be sets. A (dependent) \emph{polynomial} (of sets)
  from $I$ to $O$ consists of a set of shapes $S$, a set of positions
  $P$ and functions $f : P \to I$, $g : P \to S$ and $h : S \to O$. The \emph{extension}
  of a polynomial is a functor $\Set/I \to \Set/O$ defined as the
  composite
  \[\begin{tikzcd}
	    {\Set/I} & {\Set/P} & {\Set/S} & {\Set/O}
	    \arrow["{f^*}", from=1-1, to=1-2]
	    \arrow["{\Pi_g}", from=1-2, to=1-3]
	    \arrow["{\Sigma_h}", from=1-3, to=1-4]
  \end{tikzcd}\]
  Where $f^*$ is the pullback functor along $f$; $\Pi_g$ is the dependent
  product operation and $\Sigma_h$ is the dependent sum operation, which
  are, respectively, the right and left adjoint of their pullback functors
  $g^*$ and $h^*$. A dependent polynomial functor from $I$ to $O$ is a functor $\Set/I
  \to \Set/O$ that is naturally isomorphic to the extension of a
  dependent polynomial from $I$ to $O$.
\end{definition}
With this interpretation of $F$ as a polynomial functor, $\el(F)$ and
$\map(F)$ are just interpreted as the action of the functor on objects
and morphisms, respectively. We interpret the constructors
$\mathsf{K},\Var$, etc. on functors in the obvious way that matches
the definitional behavior of $\el$ and $\map$. The non-trivial part of
the construction is verifying that such such constructions are closed
under being polynomial. The details are tedious but straightforward
extension of prior work on dependent polynomials and indexed
containers and we have verified the construction in Agda.

We use dependent polynomials functors on sets as these are guaranteed
to have initial algebras. Further, these initial algebras are readily
constructed in our Agda implementation as an inductive type of $IW$
trees which are already available in the Cubical library of Agda
\cite{The_Agda_Community_Cubical_Agda_Library_2024}.  Then an element $F \in \sem{X \to
  \SPF\,X}\gamma$ is an $\sem{X}\gamma$-indexed family of polynomial
functors from $\sem{X}\gamma\times \Sigma^*$ to $\Sigma^*$, and taking
the product of these constructs a polynomial functor from
$\sem{X}\gamma\times \Sigma^*$ to itself. Then $\sem{\mu\,F}\gamma$ is
defined to be the initial algebra of this functor, and the initial
algebra structure is used to interpret $\roll, \fold$ and the
corresponding axioms.

The remaining details of the interpretation of linear terms as parse
transformers is a relatively straightforward extension of existing semantics for
linear logic in monoidal categories \cite{seely89}. In
\ifarxiv{\cref{sec:denotational}}\else{the extended version of the paper}\fi, we include the denotations
of linear terms and prove that our semantics respects the equational theory of
\theoryabbv --- as well as \cref{ax:dist,ax:disjointness,ax:string-top}.

\subsection{Agda Implementation \Agda}
\label{subsec:impl}
This denotational semantics in grammars and parse transformers serves
as the basis for our prototype implementation of \theoryabbv in
Cubical Agda \cite{schaefer_2025_15243560}. The implementation is a shallow embedding, meaning that
rather than formalizing a syntax of \theoryabbv types and terms, we interpret
the non-linear universe $U$ as Cubical Agda's universe of (homotopy) sets (at
some universe level $\ell$), $\hSet~\ell$, and the linear universe $L$ as the
function type
$\StringSem \to \hSet~\ell$. We use $\hSet~\ell$ as it ensures that uniqueness of identity proofs holds for the interpretation of all types in \theoryabbv, as it does in any extensional type theory. Then we implement each of the type and
term constructors of \theoryabbv as combinators on formal grammars or
parse transformers, and use Cubical Agda's equality type to model the
term equalities. Cubical Agda is convenient for this purpose as it has
built-in support for function extensionality which is convenient for the verification of equality rules. Axioms such as distributivity of $\oplus/\&$ and
disjointness of constructors are then provable directly in Agda, and
we are careful to only construct grammars, terms and proofs using
constructs from \theoryabbv.

The main difference
between our shallow embedding and \theoryabbv is that our linear terms
are written in a combinator-style, without being able to use named
variables in linear terms. For example, the function $h$ from
\cref{fig:kleeneabstractproof} would be written with the fold combinator applied
to subexpressions for the $nil$ and $cons$ cases,
\(
h = fold~nil~(cons \circ id \otimes cons \circ assoc^{-1})
\). Notice in the $cons$ case of the fold, we manually reassociate
$(\literal{a} \otimes \literal{a}) \otimes \literal{a}^{*}$ to
$\literal{a} \otimes (\literal{a} \otimes \literal{a}^{*})$, then act on the
left by the identity and on the right by $cons$ to get a parse
$\literal{a} \otimes \literal{a}^{*}$, and finally we act again by $cons$ to
finally produce a parse of $\literal{a}^{*}$.
Even in the small example from
\cref{fig:kleeneabstractproof} we must manage additional complexity --- such as
manually reassociations --- and
this only grows more complex in larger programs. On the other hand, a benefit of this shallow embedding is that
the parsers are immediately available to a larger Agda development, as
they are just normal Agda code. In future work we will look to
implement a type checker for a syntax closer to the presentation in
this paper, while maintaining the easy integration with an
existing proof assistant.

\section{Related and Future Work}
\label{sec:discussion}

\subsection{Related Work}
\paragraph{Grammars as (Linear) Types}
Lambek's original syntactic calculus \cite{lambek58} describes a
logical system for linguistic derivations, and it can be given
semantics inside of a non-commutative biclosed monoidal category
\cite{lambek1988categorial}. This led to many uses of non-commutative
linear logic and lambda calculi in linguistics
\cite{buszkowskiTypeLogicsGrammar2003} --- including mechanized
categorial grammar parsers \cite{Guillaume2024,ranta-2011}. This style
of grammar formalism has gone by the names Lambek calculus or
\emph{categorial grammar}, and it is equal in expressivity to
context-free grammars.  The existing works on categorial grammar are
different in nature to our approach: they are based on non-commutative
linear logic, but their terms do not include elimination rules, and so
can only express parse trees, not verified parsers.

The most similar prior work to our own is Luo's Lambek calculus with
dependent types \cite{luo}. They present two systems: one like ours where linear
types may only depend on non-linear ones, and another that
allows linear types to depend on other linear types and supports
directed $\Pi$ and $\Sigma$ types that have no analog in our
system. They do not provide a semantics for these connectives, and it
is unclear how to interpret their connectives in our grammar
semantics. Further, we provide several examples showing that
\theoryabbv is a practical system for describing grammar formalisms
and parsers, and it is unclear if these could be implemented in their calculus.

Frisch and Cardelli introduced the use of a simple type system to reason about
regular expressions up to weak equivalence \cite{frischCardelli}.
Henglein and Nielsen build on this type-theoretic view of grammars, pointing out that the values of the types correspond precisely to the parse trees \cite{henglein_regular_2011}.
\theoryabbv extends this view of grammars as types to a much broader class of grammars, as well as providing a syntax and equational theory for \emph{parse transformers}, showing that the terms can be viewed not just as parse trees, but as intrinsically verified parsers.

\citet{elliottSymbolicAutomaticDifferentiation2021} uses essentially the same denotational semantics as ours, interpreting regular expressions as type-valued functions on strings. Our denotational semantics of \theoryabbv extends this to a much broader class of grammars.

\paragraph{Dependent Linear Types}
Our syntax for non-commutative dependent linear type theory is based on the
dependent commutative linear type theory of
\citet{krishnaswami_integrating_2015}, itself an extension of Benton's
linear-non-linear calculus \cite{benton1994}. A distinct feature of these
systems is that linear types and non-linear types are distinct sorts, so the
linear logic $!$ operator is not a primitive, but is instead definable, for
instance in \theoryabbv as $!A = \LinSigTy{\_}{\ltonl A}{I}$.  V\'ak\'ar
develops a dependent linear type theory where instead the non-linear types are
accessed using the $!$ modality, similar to Girard's original approach
\cite{vakar2015,girard_linear_1987}. Additionally, V\'ak\'ar develops a general
categorical semantics for their system, whereas we have only developed a single
denotational model. It should be straightforward to adapt V\'ak\'ar's general
semantics approach to a non-commutative variant that would apply to
\theoryabbv. This may have applications in finding alternative models, or
developing logical relations proofs categorically.

\paragraph{Relation to Separation Logic}

\theoryabbv is similar in spirit to separation logic
\cite{reynolds_separation_2002}. Semantically, they are closely
related: linear types in \theoryabbv denote families of sets indexed
by a monoid of strings, whereas separation logic formulae typically
denote families of predicates indexed by an ordered partial
commutative monoid of worlds \cite{jung_higher-order_2016}. The
monoidal structure in both cases is are instances of the
category-theoretic notion of Day convolution monoidal structure
\cite{day1970construction}. From a separation-logic perspective, our
notion of memory is very primitive: a memory shape is just a string of
characters and the state of the memory is never allowed to evolve.

This semantic connection to separation logic suggests an avenue of
future work: to develop a program logic based on non-commutative
separation logic for verifying imperative implementations of
parsers. This could be implemented by modifying an existing separation
logic implementation or embedding the logic within \theoryabbv.

\subsection{Future Work}
In this work we have demonstrated feasibility of \theoryabbv for the
verification of formal grammar theory and sound-by-construction parsers. In
future work, we aim to extend \theoryabbv to be a practical tool for developing
verified parsing components of larger verified software systems.

\paragraph{Verified Parser Generators}
We aim to extend our work on parsing to verify other algorithms. For instance,
our regular expression parser makes somewhat arbitrary choices
when the grammar is ambiguous. In the future we aim to verify Frisch and
Cardelli's \emph{greedy} algorithm for regular expression
disambiguation. Defining this algorithm in \theoryabbv would provide soundness
by construction, but it is not obvious if the greediness property could be
verified easily as well.

We also aim to adapt the approach used for our context-free grammar parsers to
$\LLL$ and $\LRR$/$\LALRR$ parsers, with the aim of developing a shared library
of intrinsically verified parsing utilities.
We would also like to investigate how high-performance verified parsers
can be implemented in \theoryabbv. This might be done by developing an efficient
compiler for \theoryabbv directly or by developing a verified compiler to an
imperative system within \theoryabbv itself.

\paragraph{Semantic Actions}

Our verification has mainly focused on the verification that a parser outputs a
correct concrete syntax tree for a grammar. However, in practice, parsers are
combined with a \emph{semantic action} that emits an \emph{abstract} syntax tree
that omits superfluous syntactic details that are unneeded in later stages of
the overall program. We can define a semantic action in \theoryabbv for a linear
type $A$ with semantic outputs in a non-linear type $X$ to be a function $\ltonl
{(A \lto \LinSigTy{\_}{X}{\top})}$. That is, a semantic action is a function that
produces a semantic element of $X$ from the concrete parses of $A$. In future
work, we aim to study the question of verifying efficient implementations of
parsers with semantic actions, and integrating them into larger verified
systems.

\paragraph{Implementation}

Our Agda prototype implementation serves as a useful proof of concept
for showing what can be implemented in \theoryabbv, but it has
downsides we aim to address in future work. Firstly, it would be
preferable to work with the more intuitive type theoretic syntax we
have used in this work, rather than the combinator-style our shallow
embedding requires. Additionally, Agda itself does not have a
high-performance implementation, and so the parsers we implement in
Agda do not have competitive performance to industry parser
generators. In future work we aim to study if we can embed a proof of
the correctness of a parser generator that produces imperative
programs, and if the correctness of those imperative
programs can be proven within \theoryabbv.

\paragraph{Type Checking and Semantic Analysis}
Our focus in this work has been on the verification of parsers for
grammars over strings, but because \theoryabbv allows for the
definition of arbitrarily powerful grammars, the system could also be
used in principle for more sophisticated semantic analysis such as
scope checking or type checking. Alternatively, we could more directly
encode type type systems as linear types in a modified version of
\theoryabbv where linear types are not grammars over strings, but
\emph{type systems} over trees. This could analogously serve as a
framework for verified type checking and static analysis.

\newpage

\begin{acks}
  This material is based upon work supported by the Air Force Office of
  Scientific Research under Grant No.~\grantnum{}{FA9550-23-1-0760}, the ERC
  Consolidator Grant BLAST and the ARIA programme on Safeguarded AI. Any
  opinions, findings, and conclusions or recommendations expressed in this
  material are those of the authors and do not necessarily reflect the views of
  the U.S. Department of Defense, the European Research Council or the Advanced
  Research and Invention Agency.
\end{acks}

\section*{Data-Availability Statement}
The Agda implementation of \theoryabbv and the formalization of our proofs are archived on Zenodo \cite{schaefer_2025_15243560} and the source code is available \href{https://github.com/maxsnew/grammars-and-semantic-actions/tree/PLDI25Artifact}{on GitHub}.

\bibliography{refs.bib}

\newpage
\ifarxiv
\appendix
\section{Syntax}
\label{sec:syntax}

\begin{figure}
  \begin{align*}
    \el\left(\Var\,M\right) B &= B\,M\\
    \el\left(\mathsf{K}\,A\right) B &= A\\
    \el\left(\bigoplus A\right) B &= \bigoplus_{y:Y}\el(A y)B\\
    \el\left(\bigamp A\right) B &= \bigamp_{y:Y}\el(A y)B\\
    \el\left(A \otimes A'\right) B &= \el(A)B \otimes \el(A')B \\
    \map\left(\Var\,M\right)\,f &= f\,M\\
    \map\left(\K\,A\right)\,f &= \lambda a. a\\
    \map\left(\bigoplus A\right)\,f &= \lambda a. \letin{\oplusinj y {a_y}}{a}{\oplusinj y {\map(A\,y)\,f\,a_y}}\\
    \map\left(\bigwith\,A\right)\,f &= \lambda a. \withlamb{y}{\map(A\,y)\,f\,(\withprj y a)}\\
    \map\left(A \otimes A'\right)\,f &= \lambda b. \letin{(a,a')}{b}{(\map(A)\,f\,a,\map(A')\,f\,a')}
  \end{align*}
  \caption{Strictly positive functors functorial actions}
  \label{fig:spf-act}
\end{figure}

In this section we include the elided syntactic forms, as well as
definitions and basic properties of linear and non-linear
substitution.

\begin{itemize}
\item In \cref{fig:spf-act}, we provide the functorial actions of the $\el$
  and $\map$ operations used to define the indexed inductive types in
  \cref{fig:iilt}.
\item In \cref{fig:small-nonlin-ty,fig:small-lin-ty} we give the inference rules
  for the smallness judgments on nonlinear types and linear types, respectively.
\item In \cref{fig:full-non-linear-types}, we include the full set of rules for
  non-linear types, and in \cref{fig:jdg-eq-nonlinear} we provide their
  judgmental equalities.
\item In \cref{fig:jdgeq}, we give the judgmental equalities between
  linear terms in \theoryabbv.
\end{itemize}

\subsection{Non-linear Types}
\label{subsec:nonlinty}
We define sum types, list types, and $Fin~n$ from primitive non-linear types.

For $X , Y : U$, define the sum type $X + Y$ as,

\[
X + Y = \SigTyLimit{b}{Bool}{elim_{Bool}(U, X, Y)(b)}
\]

For $n : Nat$, define $Fin~n$ as,

\begin{align*}
  Fin~0 &= \bot \\
  Fin~(suc~n) &= 1 + (Fin~n)
\end{align*}

For $X : U$, define $List~X$ as,

\[
  List~X = \SigTyLimit{n}{Nat}{\PiTyLimit{k}{Fin~n}{X}}
  \]

\begin{figure}
  \footnotesize
  \begin{mathpar}
    \inferrule{~}{\Gamma \vdash 1 \isSmall}

    \inferrule{~}{\Gamma \vdash Bool \isSmall}

    \inferrule{~}{\Gamma \vdash \bot \isSmall}

    \inferrule{~}{\Gamma \vdash Nat \isSmall}

    \inferrule{\Gamma \vdash X \isSmall \and \Gamma,x:X \vdash Y \isSmall}{\Gamma \vdash \PiTyLimit {x}{X}{Y} \isSmall}

    \inferrule{\Gamma \vdash X \isSmall \and \Gamma,x:X \vdash Y \isSmall}{\Gamma \vdash \SigTyLimit {x}{X}{Y} \isSmall}

    \inferrule{\Gamma \vdash M : U}{\Gamma \vdash \unquoteTy M \isSmall}

    \inferrule{\Gamma \vdash A \isSmallLin}{\Gamma \vdash {\ltonl A} \isSmall}

    \inferrule{\Gamma \vdash X \isSmall \and \nonlinterm \Gamma M X \and
      \nonlinterm \Gamma N X}{\Gamma \vdash {M =_{X} N} \isSmall}
  \end{mathpar}
  \caption{Small Non-linear Types}
  \label{fig:small-nonlin-ty}
\end{figure}

\begin{figure}
  \footnotesize
  \begin{mathpar}
    \inferrule{~}{\Gamma \vdash I \isSmallLin}

    \inferrule{c \in \Sigma}{\Gamma \vdash {\literal{c}} \isSmallLin}

    \inferrule{\Gamma \vdash A \isSmallLin \and \Gamma \vdash B \isSmallLin}{
      \Gamma \vdash {A \otimes B} \isSmallLin}

    \inferrule{\Gamma \vdash A \isSmallLin \and \Gamma \vdash B \isSmallLin}{
      \Gamma \vdash {A \lto B} \isSmallLin}

    \inferrule{\Gamma \vdash A \isSmallLin \and \Gamma \vdash B \isSmallLin}{
      \Gamma \vdash {A \tol B} \isSmallLin}

    \inferrule{\Gamma \vdash X \isSmall \and \Gamma,x:X \vdash A \isSmallLin}{\Gamma \vdash \LinSigTyLimit {x}{X}{A} \isSmallLin}

    \inferrule{\Gamma \vdash X \isSmall \and \Gamma,x:X \vdash A \isSmallLin}{\Gamma \vdash \LinPiTyLimit {x}{X}{A} \isSmallLin}

    \inferrule{\Gamma \vdash A \isSmallLin \and
      \Gamma \vdash B \isSmallLin \and
      \nonlinterm \Gamma f {\ltonl{(A \lto B)}} \and
      \nonlinterm \Gamma g {\ltonl{(A \lto B)}}}{\Gamma \vdash {\equalizer
        {a}{f}{g}} \isSmallLin}

    \inferrule{\Gamma \vdash M : L}{\Gamma \vdash \unquoteTy M \isSmallLin}
  \end{mathpar}
  \caption{Small Linear Types}
  \label{fig:small-lin-ty}
\end{figure}

\begin{figure}
  \begin{mathpar}
  \footnotesize
    \boxed{\inferrule{\ctxwffjdg \Gamma X}{\nonlinterm \Gamma M X}}

    \inferrule{~}{\nonlinterm {\Gamma , x : X , \Gamma'} {x} {X}}

    \inferrule{\nonlinterm \Gamma M Y \and \ctxwffjdg \Gamma {X \equiv Y}}{\nonlinterm {\Gamma} {M} {X}}

    \inferrule{~}{\nonlinterm {\Gamma} {tt} {1}}

    \inferrule
    {\nonlinterm \Gamma M \bot \and \ctxwffjdg \Gamma X}
    {\nonlinterm \Gamma {elim_{\bot}(X, M)} {X}}

    \inferrule{~}{\nonlinterm \Gamma {true} {Bool}}

    \inferrule{~}{\nonlinterm \Gamma {false} {Bool}}

    \inferrule
    {
      \ctxwffjdg {\Gamma , b : Bool} {X(b)} \and
      \nonlinterm \Gamma {M_{0}} {X(false)} \and
      \nonlinterm \Gamma {M_{1}} {X(true)}
    }
    {\nonlinterm
        {\Gamma, b : Bool}
        {elim_{Bool}(X, M_0, M_1)(b)}
        {X(b)}
    }

    \\

    \inferrule{~}{\nonlinterm {\Gamma} {0} {Nat}}

    \inferrule{\nonlinterm {\Gamma} {n} {Nat}}{\nonlinterm {\Gamma} {suc~n} {Nat}}

    \inferrule{\ctxwffjdg {\Gamma , n : Nat} {X(n)} \and
           \nonlinterm {\Gamma} {M_0} {X(0)} \and
           \nonlinterm {\Gamma, n : Nat, x : X(n)} {M_{suc}(n , x)} {X(suc~n)}}
         {\nonlinterm
           {\Gamma, n : Nat}
           {elim_{Nat}(X,M_0,M_{suc})(n)}
           {X(n)}}

    \\

    \inferrule{\nonlinterm \Gamma M X \and \nonlinterm \Gamma {N} {\subst Y {M}
        {x}}}{\nonlinterm {\Gamma} {(M , N)} {\SigTyLimit {x} {X} {Y}}}

    \inferrule{\nonlinterm {\Gamma} {M} {\SigTyLimit {x} {X} {Y}}}
              {\nonlinterm {\Gamma} {M.fst} X}

    \inferrule{\nonlinterm {\Gamma} {M} {\SigTyLimit {x} {X} {Y}}}
              {\nonlinterm {\Gamma} {M.snd} {\subst Y {M.fst} {x}}}

    \inferrule{\ctxwffjdg \Gamma {\PiTyLimit {x} {X} {Y}} \and \nonlinterm
      {\Gamma , x : X} {M} {Y}}
              {\nonlinterm {\Gamma} {\lamb {x} {M}} {\PiTyLimit {x} {X} {Y}}}

    \inferrule{\nonlinterm {\Gamma} {M} {\PiTyLimit {x} {X} {Y}} \and
      \nonlinterm \Gamma N X}
              {\nonlinterm \Gamma {M~N} {\subst Y {N} {x}}}

    \inferrule{\ctxwffjdg \Gamma X}{\nonlinterm \Gamma {\quoteTy X} U}

    \inferrule{\linctxwffjdg \Gamma A}{\nonlinterm \Gamma {\quoteTy A} L}

    \inferrule{\linterm \Gamma \cdot e A}{\nonlinterm \Gamma e \ltonl{A}}

  \end{mathpar}
  \caption{Non-linear Typing}
  \label{fig:full-non-linear-types}
\end{figure}

\begin{figure}
  \footnotesize
  \begin{mathpar}
    \boxed{
      \inferrule
      {\nonlinterm \Gamma M X \and
       \nonlinterm \Gamma N X}
      {\nonlinterm \Gamma {M \equiv N} {X}}}

  \inferrule{\nonlinterm \Gamma M 1 \and \nonlinterm \Gamma N 1}{\nonlineq \Gamma M N 1}


  \inferrule{\nonlinterm {\Gamma , x : X} {M} {Y} \and \nonlinterm {\Gamma} {N} {X}}{\nonlineq \Gamma {\left(\lamb {x} {M} \right)~N} {\subst M {N}
      {x}} {X}}

  \inferrule{\nonlinterm \Gamma M {\PiTyLimit {x}{X}{Y}}}{\nonlineq \Gamma {M} {\lamb {x} {M~x}} {\PiTyLimit{x}{X}{Y}}}

  \inferrule{\nonlinterm \Gamma M {X} \and \nonlinterm {\Gamma , x : X} {N}
    {\subst Y {M} {x}}}{\nonlineq \Gamma {(M , N).fst} {M} {X}}

  \inferrule{\nonlinterm \Gamma M {X} \and \nonlinterm {\Gamma , x : X} {N}
    {\subst Y {M} {x}}}{\nonlineq \Gamma {(M , N).snd} {N} {\subst Y {M} {x}}}

  \inferrule{\nonlinterm \Gamma {M} {\SigTy {x} {X} {Y}}}{\nonlineq \Gamma {(M.fst , M.snd)} {M} {\SigTyLimit{x}{X}{Y}}}

  \inferrule
  {
    \nonlinterm \Gamma {M_{0}} {X(false)} \and
    \nonlinterm \Gamma {M_{1}} {X(true)}
  }
  {\nonlineq
      {\Gamma}
      {elim_{Bool}(X, M_0,M_1)(false)}
      {M_0}
      {X(b)}
  }

  \inferrule
  {
    \nonlinterm \Gamma {M_{0}} {X(false)} \and
    \nonlinterm \Gamma {M_{1}} {X(true)}
  }
  {\nonlineq
      {\Gamma}
      {elim_{Bool}(X, M_0,M_1)(true)}
      {M_1}
      {X(b)}
  }

  \inferrule{\nonlinterm {\Gamma} {M_0} {X(0)} \and
      \nonlinterm {\Gamma, n : Nat, x : X(n)} {M_{suc}(n , x)} {X(suc~n)}}
      {\nonlineq {\Gamma} {elim_{Nat}(X,M_0,M_{suc})(0)} {M_0} {X(0)}}

  \inferrule{\nonlinterm {\Gamma} {M_0} {X(0)} \and
      \nonlinterm {\Gamma, n : Nat, x : X(n)} {M_{suc}(n , x)} {X(suc~n)}
    }
      {\nonlineq {\Gamma, n : Nat} {elim_{Nat}(X,M_0,M_{suc})(suc~n)} {M_{suc}(n ,
          elim_{Nat}(X,M_0,M_{suc})(n))} {X(suc~n)}}

  \inferrule{\nonlinterm \Gamma M U}{\nonlineq \Gamma {\quoteTy {\unquoteTy {M}}}
  {M} {U}}

  \inferrule{\nonlinterm \Gamma M L}{\nonlineq \Gamma {\quoteTy {\unquoteTy {M}}}
  {M} {L}}
  \end{mathpar}
  \caption{Judgmental equality for non-linear terms}
  \label{fig:jdg-eq-nonlinear}
\end{figure}

\begin{figure}
  \footnotesize
  \begin{mathpar}
    \boxed{
      \inferrule
      {\linterm \Gamma \Delta e A \and
       \linterm \Gamma \Delta e' A}
      {\Gamma ; \Delta \vdash e \equiv e' : A}}

%
%
    \inferrule{\Gamma ; \Delta , a : A \vdash e : C \\ \Gamma ; \Delta' \vdash e' : A }{\Gamma; \Delta,\Delta' \vdash \app {(\lamblto a e)} {e'} \equiv \subst e {e'} a : C}
    \and
    \inferrule{\Gamma ; \Delta \vdash e : A \lto B}{\Gamma; \Delta \vdash e \equiv \lamblto a {\app e a} : A \lto B}
    \and
    \inferrule{\Gamma ; a : A , \Delta \vdash e : C \\ \Gamma ; \Delta' \vdash e' : A }{\Gamma; \Delta , \Delta' \vdash \apptol {(\lambtol a e)} {e'} \equiv \subst e {e'} a : C}
    \and
    \inferrule{\Gamma ; \Delta \vdash e : B \tol A}{\Gamma; \Delta \vdash e \equiv \lambtol a {\apptol e a} : B \tol A}
    \and
    \inferrule{\Gamma , x : X \vdash e : A \\ \Gamma \vdash M : X}{\Gamma; \Delta \vdash \app {(\dlamb x e)} {M} \equiv \subst {e} M x : C}
    \and
    \inferrule{\Gamma ; \Delta \vdash e : \LinPiTyLimit {x} {X} {A}}{\Gamma; \Delta \vdash e \equiv \dlamb x {\app e x} : \LinPiTyLimit x X A}
    \and
%
%
%
    \and
    \inferrule{\Gamma ; \Delta_1 , \Delta_2 \vdash e : C}{\Gamma; \Delta_1 , \Delta_2 \vdash \letin {()} {()} e \equiv e : C}
    \and
    \inferrule{\Gamma ; \Delta_2 \vdash e : I \\ \Gamma ; \Delta_1 , a : A , \Delta_3 \vdash e' : C}{\Gamma; \Delta_1 , \Delta_3 \vdash \letin {()} e {\subst {e'} {()} a} \equiv \subst {e'} e a : C}
    \and
    \inferrule{\Gamma ; \Delta_2 \vdash e : A \\ \Gamma ; \Delta_3 \vdash e' : B \\ \Gamma ; \Delta_1 , a : A , b : B , \Delta_4 \vdash e'' : C}{\Gamma; \Delta_1 , \Delta_2, \Delta_3 , \Delta_4 \vdash \letin {(a, b)} {(e , e')} e'' \equiv e'' \{ e/a , e'/b \} : C}
    \and
    \inferrule{\Gamma ; \Delta_2 \vdash e : A \otimes B \\ \Gamma ; \Delta_1 , c : A \otimes B , \Delta_3 \vdash e' : C }{\Gamma; \Delta_1, \Delta_2, \Delta_3 \vdash \letin {(a , b)} e {\subst {e'} {(a, b)} c} \equiv \subst {e'} e c : C}
    \and
    \inferrule{\Gamma \vdash M : X \\ \Gamma ; \Delta_2 \vdash e : A \\ \Gamma , x : X \vdash \Delta_1 , a : A , \Delta_3 }{\Gamma;\Delta_1 , \Delta_2 , \Delta_3 \vdash \letin {\sigma~x~a} {\sigma~M~e} {e'} \equiv e' \{ M/x , e/a \} : C}
    \and
    \inferrule{\Gamma ; \Delta_1, y : \LinSigTyLimit {x} {X} {A}, \Delta_2 \vdash e' : C \\ \Gamma ; \Delta_2 \vdash e : \LinSigTyLimit {x} {X} {A}}{\Gamma; \Delta_1 , \Delta_2 , \Delta_3 \vdash \letin {\sigma~x~a} e {\subst {e'} {\sigma~x~a} y} \equiv \subst {e'} e y : C}
    \and
    \inferrule
    {\Gamma ; \Delta \vdash e : A \\ \Gamma ; \Delta \vdash f~e \equiv g~e}
    {\Gamma ; \Delta \vdash \equalizerpi {\equalizerin {e}} \equiv e : A}
    \and
    \inferrule
    {\Gamma ; \Delta \vdash e : \equalizer e f g}
    {\Gamma ; \Delta \vdash \equalizerin {\equalizerpi {e}} \equiv e :
      \equalizer e f g}
\end{mathpar}
  \caption{Judgmental equality for linear terms}
  \label{fig:jdgeq}
\end{figure}

\subsection{Substitutions}
\begin{definition}
  The set of (non-linear) substitutions $\gamma \in
  \textrm{Subst}(\Gamma,\Gamma')$ where $\Gamma \isCtx$ and $\Gamma'
  \isCtx$ is defined by recursion on $\Gamma$:
  \begin{align*}
    \textrm{Subst}(\Gamma,\cdot) &= \{\cdot \}\\
    \textrm{Subst}(\Gamma,\Gamma',x:A) &= \{ (\gamma,M/x) \pipe \gamma \in \textrm{Subst}(\Gamma,\Gamma') \wedge \Gamma \vdash M : A[\gamma] \}
  \end{align*}
  simultaneously with an action of substitution on types, terms,
  etc. in the standard way.

  It is straightforward, but laborious to establish that all forms in
  the type theory that are parameterized by a non-linear context
  $\Gamma$ support the admissible actions of a substitution $\gamma
  \in \Subst(\Gamma',\Gamma)$ given in \cref{fig:non-linear-substitution,fig:substact}.
\end{definition}

\begin{figure}
  \footnotesize
  \begin{mathpar}
    \inferrule
    {\ctxwffjdg \Gamma X}
    {\ctxwffjdg {\Gamma'} {X[\gamma]}}

    \inferrule
    {\Gamma \vdash X \isSmall}
    {\Gamma' \vdash X[\gamma] \isSmall}

    \inferrule
    {\ctxwffjdg \Gamma {X \equiv Y}}
    {\ctxwffjdg {\Gamma'} {X[\gamma] \equiv Y[\gamma]}}

    \inferrule
    {\nonlinterm \Gamma {M} {X}}
    {\nonlinterm {\Gamma'} {M[\gamma]} {X[\gamma]}}

    \inferrule
    {\nonlineq \Gamma {M} {N} {X}}
    {\nonlineq {\Gamma'} {M[\gamma]} {N[\gamma]} {X[\gamma]}}

    \inferrule
    {\Gamma \vdash \Delta \isLinCtx}
    {\Gamma' \vdash \Delta[\gamma] \isLinCtx}

    \inferrule
    {\Gamma \vdash A \isLinTy}
    {\Gamma' \vdash A[\gamma] \isLinTy}

    \inferrule
    {\linctxwffjdg \Gamma {A \equiv B}}
    {\linctxwffjdg {\Gamma'} {A[\gamma] \equiv B[\gamma]}}

    \inferrule
    {\Gamma;\Delta \vdash e : A}
    {\Gamma'; \Delta[\gamma] \vdash e[\gamma] : A[\gamma]}

    \inferrule
    {\Gamma;\Delta \vdash e \equiv f : A}
    {\Gamma'; \Delta[\gamma] \vdash e[\gamma] \equiv f[\gamma] : A[\gamma]}
  \end{mathpar}
  \caption{Non-linear substitution}
  \label{fig:non-linear-substitution}
\end{figure}

\begin{definition}
  Let $\Gamma \vdash \Delta \isLinCtx$ and $\Gamma \vdash \Delta'
  \isLinCtx$. The set of linear substitutions $\Subst(\Delta',\Delta)$
  is defined by recursion on $\Delta$:
  \begin{align*}
    \Subst(\Delta',\cdot) &= \{ \cdot \pipe \Delta' = \cdot \}\\
    \Subst(\Delta',(\Delta,a:A)) &= \{ (\delta, e/a) \pipe \delta \in \Subst(\Delta_1, \Delta) , \Delta_2 \vdash e : A , \Delta' = (\Delta_1,\Delta_2)\}
  \end{align*}

  Given substitutions $\delta_1 \in \Subst(\Delta_1', \Delta_1)$ and
  $\delta_2 \in \Subst(\Delta_2', \Delta_2)$, we can define a
  substitution $\delta_1,\delta_2 \in
  \Subst((\Delta_1',\Delta_2'),(\Delta_1,\Delta_2))$. Furthermore, for
  any substitution $\delta \in \Subst(\Delta,(\Delta_1,\Delta_2))$, we
  can deconstruct $\delta = \delta_1,\delta_2$ with $\delta_1 \in
  \Subst(\Delta_1', \Delta_1)$ and $\delta_2 \in \Subst(\Delta_2',
  \Delta_2)$.
\end{definition}

\begin{figure}
  \footnotesize
  \begin{align*}
    a[e/a] &= e\\
    (e_1,e_2)[\delta_1,\delta_2] &= (e_1[\delta_1], e_2[\delta_2])\\
    (\letin {(a , b)} e {e'})[\delta_1,\delta_2,\delta_3] &= \letin{(a,b)} {e[\delta_2]}{e'[\delta_1,a/a,b/b,\delta_2]}\\
    ()[\cdot] &= ()\\
    \letin{()} e {e'}[\delta_1,\delta_2,\delta_3] &= \letin{()} {e[\delta_2]}{e'[\delta_1,a/a,b/b,\delta_2]}\\
    (\lamblto {a} e)[\delta] &= \lamblto a {e[\delta,a/a]}\\
    (\applto {e'} {e})[\delta_1,\delta_2] &= \applto {e'[\delta_1]} {e[\delta_2]}\\
    (\lamblto {a} e)[\delta] &= \lamblto a {e[\delta,a/a]}\\
    (\applto {e'} {e})[\delta_1,\delta_2] &= \applto {e'[\delta_1]} {e[\delta_2]}\\
    (\lambtol {a} e)[\delta] &= \lamblto a {e[a/a,\delta]}\\
    (\apptol {e'} {e})[\delta_1,\delta_2] &= \apptol {e'[\delta_1]} {e[\delta_2]}\\
    (\dlamb x e)[\delta] &= \dlamb x {e[\delta]}\\
    (e\,.\pi\,M)[\delta] &= (e[\delta]\,.\pi\,M)\\
    (\sigma\,M\,e)[\delta] &= \sigma\,M\,e[\delta]\\
    (\letin{\sigma\,x\,a}{e}{e'})[\delta_1,\delta_2,\delta_3] &= \letin{\sigma\,x\,a} {e[\delta_2]}{e'[\delta_1,\delta_3]}\\
    (\equalizerin{e})[\delta] &= \equalizerin{e[\delta]}\\
    (\equalizerpi{e})[\delta] &= \equalizerpi{e[\delta]}
  \end{align*}
\caption{Action of substitution on linear terms}
\label{fig:substact}
\end{figure}

\begin{definition}
  Given any $\Gamma ; \Delta \vdash e : A$ and $\delta \in
  \Subst(\Delta',\Delta)$, we define the action of the substitution on
  $e$ in \cref{fig:substact}, frequently using the inversion principle to split
  the substitution into constituent components.
  By induction on linear term and equality judgments, we establish the
  following admissible rules for $\delta \in \Subst(\Delta',\Delta)$:
  \begin{mathpar}
    \inferrule
    {\Gamma; \Delta \vdash e : A}
    {\Gamma; \Delta' \vdash e[\delta] : A}

    \inferrule
    {\Gamma; \Delta \vdash e \equiv f : A}
    {\Gamma; \Delta' \vdash e[\delta] \equiv f[\delta'] : A}
  \end{mathpar}
\end{definition}

\section{Denotational Semantics}
\label{sec:denotational}
Here we extend the denotational semantics from
\cref{sec:semantics-and-metatheory} to cover all of \theoryname
syntax. Here we will freely use that the category of grammars is a
complete, co-complete biclosed monoidal category, and use categorical
notation for the constructions in the denotational semantics. For
example, we will use the same notation $I,\otimes,\lto,\tol$ for the
biclosed monoidal structure of $\Grammar$ that we do for the
corresponding syntactic notions.

\begin{definition}[Denotation of Linear Contexts]
  The semantics of linear contexts $\Gamma \vdash \Delta \isLinCtx$ is
  defined as follows:
  \max{make sure this aligns with our interpretation of $\lto$/$\tol$}
  \begin{align*}
    \sem{\cdot}\,\gamma &= I \\
    \sem{\Delta,x:A}\,\gamma &= \sem{\Delta}\,\gamma \otimes \sem{A}\gamma \\
  \end{align*}
\end{definition}

\begin{definition}[Denotation of Linear Substitutions]
  The semantics of a linear substitution
  $\delta : \textrm{Subst}(\Delta' , \Delta)$ are given as maps
  $\semg{\delta} : \semg{\Delta'} \to {\Delta}$. Define $\semg{\delta}$ by
  recursion on $\delta$:
  \begin{align*}
    \semg{\cdot} &= id_{I} \\
    \semg{\delta,e/a} &= \semg{\delta} \otimes \semg{e} \circ m_{\Delta_{1},\Delta_{2}}
  \end{align*}
where $\Delta_{2} \vdash e : A$ and $\Delta' = \Delta_{1},\Delta_{2}$.
\end{definition}

\begin{theorem}
  For any $\Gamma \vdash \Delta_1,\Delta_2 \isLinCtx$ and $\gamma \in
  \sem{\Gamma}$ there is a natural isomorphism $m_{\Delta_1,\Delta_2} :
  \sem{\Delta_1, \Delta_2}\gamma \cong \sem{\Delta_1}\gamma \otimes
  \sem{\Delta_2}\gamma$.

  This can be extended to a sequence of contexts of any length.
\end{theorem}
\begin{proof}
Construct $m_{\Delta_{1},\Delta_{2}}$ by recursion on $\Delta_{2}$.
\begin{align*}
  m_{\Delta_{1},\cdot} & = \rho^{-1} \\
  m_{\Delta_{1},(\Delta_{2}, a : A)} & = \alpha \circ m_{\Delta_{1},\Delta_{2}} \otimes id
\end{align*}
\end{proof}

\begin{lemma}
  \label{lem:subst}
  For each term $\Delta \vdash e : A$ and substitution
  $\delta : \textrm{Subst}(\Delta', \Delta)$, the semantics of $\delta$ acting
  on $e$ splits into the composition $\semg{e[\delta]} = \semg{e} \circ \semg{\delta}$.
\end{lemma}

\subsection{Grammar Semantics for Linear Terms}
Here we define denotations of linear terms. Note that the denotations
interpret typing derivations, not raw terms, as the data of how
contexts are split is needed in order to construct the correct
associator functions. Further, we demonstrate that the denotational
semantics respects the equational theory of \theoryabbv. The
correctness of the equational theory heavily relies on the
\emph{coherence theorem} for monoidal categories. The coherence
theorem says that any diagram in a monoidal category constructed using
only associators $\alpha_{A,B,C} : (A \otimes B) \otimes C \cong A
\otimes (B \otimes C)$, unitors $\rho_A : A \otimes I \cong A$ and
$\lambda_A : I \otimes A \cong A$ and compositions and tensor products
of these, commutes. We call a morphism built in this way a
\emph{generalized associator}.

\subsubsection{Variables}
Note that the denotation of a singleton context $a : A$ is given as
\[
  \semg{a : A} = \semg{\cdot , a : A} = I \otimes \semg{A}
\]
So for $a: A \vdash a : A$, the denotation of a single variable term
$\semg{a} : \semg{a : A} \to \semg{A}$ is given by the left unitor
\[
  \semg{a} = \lambda
\]

\subsubsection{Linear Unit}
\paragraph{$I$-Introduction}
$\semg{()} : \semg{\cdot} \to \semg{I}$.
\[
\semg{()} = id_{I}
\]

\paragraph{$I$-Elimination}
$\semg{\letin {()} e e'} : \semg{\Delta'_{1},\Delta,\Delta'_{2}} \to \semg{C}$
defined in the following diagram,

\begin{center}
\begin{tikzcd}
  \semg{\Delta'_1,\Delta,\Delta'_2} \arrow[r , "m_{(\Delta'_1,\Delta),\Delta'_2}"] &
  \semg{\Delta'_1,\Delta} \otimes \semg{\Delta'_2}
    \arrow[r , "m_{\Delta'_1,\Delta} \otimes id"] &
  \left(\semg{\Delta'_1} \otimes \semg{\Delta} \right) \otimes \semg{\Delta'_2}
    \arrow[d , "(id \otimes \semg{e}) \otimes id"] \\
    \semg{\Delta'_1,\Delta'_2} \arrow[d , "\semg{e'}"]
  & \arrow[l , "m^{-1}_{\Delta'_1,\Delta'_2}"]
    \semg{\Delta'_1} \otimes \semg{\Delta'_2}
  & \left(\semg{\Delta'_1} \otimes I \right) \otimes \semg{\Delta'_2}
  \arrow[l , "\rho \otimes id"] \\
  \semg{C}
\end{tikzcd}
\end{center}

We demonstrate that the denotations of the introduction and elimination forms
for $I$ obey the $\beta$ and $\eta$ equalities for $I$.

\paragraph{$I\beta$}
Given
$\Delta'_{1},\cdot,\Delta'_{2} \vdash e' : C$, the desired
$\beta$ law is
\[
  \semg{\letin {()} {()} {e'}} = \semg{e}
\]
\begin{proof}
\begin{align*}
  \semg{\letin {()} {()} {e'}}
  & = \semg{e'} \circ m^{-1}_{\Delta'_{1},\Delta'_{2}} \circ \rho \otimes id \circ (id \otimes \semg{()}) \otimes id \circ m_{\Delta_{1},\cdot} \otimes id \circ m_{(\Delta'_{1},\cdot),\Delta'_{2}} \\
  & = \semg{e'} \circ m^{-1}_{\Delta'_{1},\Delta'_{2}} \circ \rho \otimes id \circ (id \otimes id) \otimes id \circ m_{\Delta_{1},\cdot} \otimes id \circ m_{\Delta'_{1},\Delta'_{2}} \\
  & = \semg{e'} \circ m^{-1}_{\Delta'_{1},\Delta'_{2}} \circ \rho \otimes id \circ \rho^{-1} \otimes id \circ m_{\Delta'_{1},\Delta'_{2}} \\
  & = \semg{e'} \tag{coherence}
\end{align*}
\end{proof}

\paragraph{$I\eta$}
Similarly, for $\Delta_{1}, a : A, \Delta_{3} \vdash e' : C$ and
$\Delta_{2} \vdash e : I$ the desired $\eta$ law is
\[
\semg{\letin {()} {e} {e'[()/a]}} = \semg{e'[e/a]}
\]
However, through application of \cref{lem:subst} it suffices to handle the case
where $e$ is a variable $a'$. That is,
\begin{align*}
\letin {()} {e} {e'[()/a]} &= (\letin {()} {a'} {e'[()/a]})[e/a'] \\
e'[e/a] &= e'[a'/a][e/a]
\end{align*}
so without loss of generality we may take may take $e$ to be variable $a'$. We
will additionally use this style of argumentation when necessary throughout this section.

\begin{proof}
\begin{align*}
  \semg{\letin {()} {a'} {e'[()/a]}}
  & = \semg{e'} \circ m^{-1}_{\Delta'_{1},\Delta'_{2}} \circ \rho \otimes id \circ (id \otimes \semg{a'}) \otimes id \circ m_{\Delta_{1},\cdot} \otimes id \circ m_{\Delta'_{1},\Delta'_{2}} \\
  & = \semg{e'} \circ m^{-1}_{\Delta'_{1},\Delta'_{2}} \circ \rho \otimes id \circ (id \otimes \lambda) \otimes id \circ m_{\Delta_{1},\cdot} \otimes id \circ m_{\Delta'_{1},\Delta'_{2}} \\
  & = \semg{e'} \tag{coherence}
\end{align*}

Because
$m^{-1}_{\Delta'_{1},\Delta'_{2}} \circ \rho \otimes id \circ (id \otimes \lambda) \otimes id \circ m_{\Delta_{1},\cdot} \otimes id \circ m_{\Delta'_{1},\Delta'_{2}}$
is a composition of generalized associators from
$\semg{\Delta'_{1},\Delta'_{2}}$ to itself, it is equal to the identity by the coherence theorem for monoidal categories.

Further by \cref{lem:subst},
\begin{align*}
  \semg{e'[a'/a]}
  &= \semg{e'} \circ \semg{a'/a} \tag{\cref{lem:subst}} \\
  &= \semg{e'} \tag{coherence}
\end{align*}
Again by the coherence theorem, $\semg{a'/a} = id$. Thus, $\eta$ law for $I$
holds in the denotational semantics.
\end{proof}

\subsubsection{Tensor}
\paragraph{$\otimes$-Introduction}
$\semg{(e_{1}, e_{2})} : \semg{\Delta,\Delta'} \to \semg{A \otimes B}$ is given
by the diagram

\begin{center}
\begin{tikzcd}[column sep = 1in]
  \semg{\Delta,\Delta'} \arrow[r, "m_{\Delta,\Delta'}"] &
  \semg{\Delta} \otimes \semg{\Delta'} \arrow[r, "\semg{e_1} \otimes \semg{e_2}"] &
  \semg{A} \otimes \semg{G}
\end{tikzcd}
\end{center}

\paragraph{$\otimes$-Elimination}
$\semg{\letin {(a,b)} e f} : \semg{\Delta'_{1},\Delta,\Delta'_{2}} \to \semg{C}$
defined via the diagram,

\begin{center}
\begin{tikzcd}
  \semg{\Delta'_1,\Delta,\Delta'_2} \arrow[r , "m_{(\Delta'_1,\Delta),\Delta'_2}"] &
  \semg{\Delta'_1,\Delta} \otimes \semg{\Delta'_2}
    \arrow[r , "m_{\Delta'_1,\Delta} \otimes id"] &
  \left(\semg{\Delta'_1} \otimes \semg{\Delta} \right) \otimes \semg{\Delta'_2}
    \arrow[d , "(id \otimes \semg{e}) \otimes id"] \\
  \semg{c} & \arrow[l , "\semg{f}"]
  \left( \left( \semg{\Delta'_1} \otimes \semg{A} \right) \otimes \semg{B}\right) \otimes \semg{\Delta'_2}
  & \left(\semg{\Delta'_1} \otimes \left( \semg{A} \otimes \semg{B}\right) \right) \otimes \semg{\Delta'_2}
    \arrow[l , "\alpha \otimes id"]
\end{tikzcd}
\end{center}

\paragraph{$\otimes\beta$}
The desired $\beta$ equality for $\otimes$ is,
\[
  \semg{\letin {(a,b)} {(a',b')} {f}} = \semg{f[a'/a,b'/b]}
\]
\begin{proof}
The left hand side reduces as follows,
\begin{align*}
  \semg{\letin {(a,b)} {(a',b')} {f}}
  & = \semg{f} \circ \alpha \otimes id \circ (id \otimes \semg{(a' , b')}) \otimes id \circ m_{\Delta,\Delta'} \\
  & = \semg{f} \circ \alpha \otimes id \circ (id \otimes \lambda \otimes \lambda \circ m_{a : A, b : B})) \otimes id \circ m_{\Delta,\Delta'} \\
  & = \semg{f} \tag{coherence}
\end{align*}

Which is equal to the right hand side,
\begin{align*}
  \semg{f[a'/a,b'/b]}
  & = \semg{f} \circ \semg{a'/a,b'/b} \\
  & = \semg{f} \tag{coherence}
\end{align*}
\end{proof}

\paragraph{$\otimes\eta$}
The desired $\eta$ equality for $\otimes$ is,
\[
  \semg{\letin {(a,b)} {c'} {f[(a,b)/c]}} = \semg{f[c'/c]}
\]
\begin{proof}
\begin{align*}
  \semg{\letin {(a,b)} {c'} {f[(a,b)/c]}}
  & = \semg{f[(a,b)/c]} \circ \alpha \otimes id \circ (id \otimes \semg{c'}) \otimes id \circ m_{\Delta,\Delta'} \\
  & = \semg{f} \circ \semg{(a,b)/c} \circ \alpha \otimes id \circ (id \otimes \semg{c'}) \otimes id \circ m_{\Delta,\Delta'} \\
  & = \semg{f} \tag{coherence}
\end{align*}

\begin{align*}
  \semg{f[c'/c]}
  &= \semg{f} \circ \semg{c'/c} \tag{\cref{lem:subst}} \\
  &= \semg{f} \tag{coherence}
\end{align*}
\end{proof}

\subsubsection{$\lto$-Functions}
\paragraph{$\lto$-Introduction}
$\semg{\lamblto a e} : \semg{\Delta} \to \semg{A \lto B}$ is defined using the
natural isomorphism
$\phi : Hom(\semg{\Delta} \otimes \semg{A}, \semg{B}) \to Hom(\semg{\Delta}, \semg{A \lto B})$
that is provided by the adjunction between $\semg{- \otimes A}$ and $\semg{A \lto -}$.
\[
\semg{\lamblto a e} = \phi\left( \semg{e} \right)
\]
\paragraph{$\lto$-Elimination}
$\semg{\applto {e'} {e}} : \semg{\Delta,\Delta'} \to \semg{B}$ is defined by the
diagram,

\begin{center}
\begin{tikzcd}[column sep=.75in]
  \semg{\Delta,\Delta'} \arrow[r , "m_{\Delta,\Delta'}"] &
  \semg{\Delta} \otimes \semg{\Delta'} \arrow[r , "id \otimes \semg{e}"] &
  \semg{\Delta} \otimes \semg{A} \arrow[r , "\phi^{-1}\left( \semg{e'} \right)"]
  &
  \semg{B}
\end{tikzcd}
\end{center}

\paragraph{$\lto\beta$}
The $\beta$ rule for $\lto$ is given by,
\[
  \semg{\applto {(\lamblto {a} {e})} {a'}} = \semg{e[a'/a]}
\]
\begin{proof}
\begin{align*}
  \semg{\applto {(\lamblto {a} {e})} {a'}}
  & = \phi^{-1}(\semg{\lamblto {a} {e}}) \circ id \otimes \semg{a'} \circ m_{\Delta,\Delta'} \\
  & = \phi^{-1}(\phi(\semg{e})) \circ id \otimes \semg{a'} \circ m_{\Delta,\Delta'} \\
  & = \semg{e} \circ id \otimes \semg{a'} \circ m_{\Delta,\Delta'} \\
  & = \semg{e} \circ id \otimes \lambda \circ m_{\Delta,\Delta'} \\
  & = \semg{e} \tag{coherence}
\end{align*}

\begin{align*}
  \semg{e[a'/a]}
  & = \semg{e} \circ \semg{a'/a} \tag{\cref{lem:subst}} \\
  & = \semg{e} \tag{coherence}
\end{align*}
\end{proof}

\paragraph{$\lto\eta$}
The $\eta$ rule for $\lto$ is given by,
\[
  \semg{\lamblto {a} {\applto e a}} = \semg{e}
\]
\begin{proof}
\begin{align*}
  \semg{\lamblto {a} {\applto e a}}
  & = \phi(\semg{\applto e a}) \\
  & = \phi(\phi^{-1}(\semg{e}) \circ id \otimes \semg{a} \circ m_{\Delta,a:A}) \\
  & = \phi(\phi^{-1}(\semg{e}) \circ id \otimes \lambda \circ m_{\Delta,a:A}) \\
  & = \phi(\phi^{-1}(\semg{e})) \tag{coherence} \\
  & = \semg{e}
\end{align*}
\end{proof}

\subsubsection{$\tol$-Functions}
\paragraph{$\tol$-Introduction}
Just as with the other linear function type, we have an adjunction between
$\semg{A \otimes -}$ and $\semg{- \tol A}$. $\semg{\lambtol a e} : \semg{\Delta} \to \semg{B \tol A}$ is defined using the
natural isomorphism
$\psi : Hom(\semg{A} \otimes \semg{\Delta}, \semg{B}) \to Hom(\semg{\Delta}, \semg{B \tol A})$
induced by this adjunction. In particular, $\semg{\lambtol a e}$ is given by $\psi$
acting on the following diagram

\begin{center}
\begin{tikzcd}[column sep=.75in]
  \semg{A} \otimes \semg{\Delta} \arrow[r , "\lambda^{-1} \otimes id"] &
  \semg{a : A} \otimes \semg{\Delta} \arrow[r , "m^{-1}_{a : A , \Delta}"] &
  \semg{a : A , \Delta} \arrow[r , "\semg{e}"] &
  \semg{B}
\end{tikzcd}
\end{center}

\paragraph{$\tol$-Elimination}
The application of a $\tol$-function,
$\semg{\apptol {e} {e'}} : \semg{\Delta',\Delta} \to \semg{B}$ defined by the diagram

\begin{center}
\begin{tikzcd}[column sep=.75in]
  \semg{\Delta',\Delta} \arrow[r , "m_{\Delta',\Delta}"] &
  \semg{\Delta'} \otimes \semg{\Delta} \arrow[r , "\semg{e'} \otimes id"] &
  \semg{A} \otimes \semg{\Delta} \arrow[r , "\psi^{-1}\left( \semg{e} \right)"] &
  \semg{B}
\end{tikzcd}
\end{center}

\paragraph{$\tol\beta$}
The $\beta$ rule for $\tol$ is given by,
\[
  \semg{\apptol {(\lambtol {a} {e})} {a'}} = \semg{e[a'/a]}
\]
\begin{proof}
\begin{align*}
  \semg{\apptol {(\lambtol {a} {e})} {a'}}
  & = \psi^{{-1}}(\semg{\lambtol {a} {e}}) \circ \semg{a'} \otimes id \circ m_{a:A,\Delta} \\
  & = \psi^{{-1}}(\semg{\lambtol {a} {e}}) \circ \lambda \otimes id \circ m_{a:A,\Delta} \\
  & = \psi^{{-1}}(\psi(\semg{e} \circ m^{-1}_{a:A,\Delta} \circ \lambda^{-1} \otimes id)) \circ \lambda \otimes id \circ m_{\Delta',\Delta} \\
  & = \semg{e} \circ m^{-1}_{a:A,\Delta} \circ \lambda^{-1} \otimes id \circ \lambda \otimes id \circ m_{a:A,\Delta} \\
  & = \semg{e}
\end{align*}

\begin{align*}
  \semg{e[a'/a]}
  & = \semg{e} \circ \semg{a'/a} \tag{\cref{lem:subst}} \\
  & = \semg{e} \tag{coherence}
\end{align*}
\end{proof}

\paragraph{$\tol\eta$}
The $\eta$ rule for $\tol$ is given by,
\[
  \semg{\lambtol {a} {\apptol e a}} = \semg{e}
\]
\begin{proof}
\begin{align*}
  \semg{\lambtol {a} {\apptol e a}}
  & = \phi(\semg{\apptol e a} \circ m^{-1}_{a:A,\Delta} \circ \lambda^{-1} \otimes id) \\
  & = \phi(\phi^{-1}(\semg{e}) \circ \semg{a} \otimes id \circ m_{a:A,\Delta} \circ m^{-1}_{a:A,\Delta} \circ \lambda^{-1} \otimes id) \\
  & = \phi(\phi^{-1}(\semg{e}) \circ \lambda \otimes id \circ m_{a:A,\Delta} \circ m^{-1}_{a:A,\Delta} \circ \lambda^{-1} \otimes id) \\
  & = \phi(\phi^{-1}(\semg{e})) \\
  & = \semg{e}
\end{align*}
\end{proof}

\subsubsection{$\bigamp$-Products}
\paragraph{$\bigamp$-Introduction}
$\semg{\dlamb x e} : \semg{\Delta} \to \prod_{x : \semg{X}} {\sem{A}(\gamma , x)}$ is
defined by the universal property of the product
\[
\semg{\dlamb x e} = \left( \sem{e}(\gamma , x) \right)_{(x : \semg{X})}
\]

\paragraph{$\bigamp$-Elimination}
$\semg{e.\pi~M} : \semg{\Delta} \to \sem{A}(\gamma , M)$ is defined using the
projection out of the product,

\begin{center}
\begin{tikzcd}
  \semg{\Delta} \arrow[r , "\semg{e}"] &
  \prod_{x : \semg{X}} {\sem{A}(\gamma , x)} \arrow[r , "\pi_M"] &
  \sem{A}(\gamma , M)
\end{tikzcd}
\end{center}

\paragraph{$\bigamp\beta$}
The $\beta$ law for $\bigamp$ is given by,
\[
  \sem{(\dlamb x e).\pi~M}(\gamma, x) = \sem{e[M/x]}(\gamma, x)
\]
\begin{proof}
\begin{align*}
  \semg{(\dlamb x e).\pi~M}
  & = \pi_{M} \circ \semg{\dlamb x e} \\
  & = \pi_{M} \circ (\sem{e}(\gamma , x))_{(x : \semg{X})} \\
  & = \sem{e}(\gamma , M)
\end{align*}
by the universal property of the product.

\begin{align*}
  \sem{e[M/x]}(\gamma , x)
  & = \sem{e}(\gamma , x) \circ  \sem{M/x}(\gamma , x) \\
  & = \sem{e}(\gamma , M)
\end{align*}
\end{proof}

\paragraph{$\bigamp\eta$}
The $\eta$ law for $\bigamp$ is given by,
\[
  \semg{(\dlamb x {e.\pi~x})} = \semg{e}
\]
\begin{proof}
\begin{align*}
  \semg{(\dlamb x {e.\pi~x})}
  & = (\semg{e.\pi~x})_{x : \semg{X}} \\
  & = (\pi_x \circ \semg{e})_{x : \semg{X}} \\
  & = \semg{e}
\end{align*}
\end{proof}
by the universal property of the product.

\subsubsection{$\bigoplus$-Sums}
\paragraph{$\bigoplus$-Introduction}
$\semg{\sigma~M~e} : \semg{\Delta} \to \coprod_{x : \semg{X}} \sem{A}(\gamma, x)$

\begin{center}
\begin{tikzcd}
  \semg{\Delta} \arrow[r , "\semg{e}"] &
  \sem{A}(\gamma, M) \arrow[r , "i_M"] &
  \coprod_{x : \semg{X}} {\sem{A}(\gamma , x)}
\end{tikzcd}
\end{center}

\paragraph{$\bigoplus$-Elimination}
$\semg{\letin {\sigma~x~a} {e} {e'}} : \semg{\Delta'_{1},\Delta,\Delta'_{2}} \to \semg{C}$
is defined in the diagram

\begin{center}
\begin{tikzcd}[column sep =large]
{\semg{\Delta'_1,\Delta,\Delta'_2}} \arrow[r, "{m_{(\Delta'_1,\Delta),\Delta'_2}}"]                                                                                                           & {\semg{\Delta'_1,\Delta} \otimes \semg{\Delta'_2}} \arrow[d, "{m_{\Delta'_1,\Delta}\otimes id}"]                             \\
{\left( \semg{\Delta'_1} \otimes \coprod_{x:\semg{X}} {\sem{A}(\gamma , x)} \right) \otimes \semg{\Delta'_2}} \arrow[d, "d"]                                                                 & \left( \semg{\Delta'_1} \otimes \semg{\Delta} \right) \otimes \semg{\Delta'_2} \arrow[l, "(id \otimes \semg{e}) \otimes id"] \\
{\coprod_{x : \semg{X}} {\left( \sem{\Delta'_1}(\gamma , x) \otimes
      \sem{A}(\gamma , x) \right) \otimes \sem{\Delta'_2}(\gamma , x)}} \arrow[r, outer sep=.15cm, "\coprod_{x:\semg{X}}({m^{-1}_{(\Delta'_1,a:A),\Delta'_2}})"] & {\coprod_{x : \semg{X}} {\sem{\Delta'_1,a : A,\Delta'_2}(\gamma , x)}} \arrow[d, "{[ \sem{e'}(\gamma , x) ]_{(x : \semg{X})}}"]     \\
\semg{C}                                                                                                                                                                                      & {\sem{C}(\gamma, x)} \arrow[l]
\end{tikzcd}
\end{center}

where $d$ is the distributivity morphism, and the last morphism implicitly
weakens $\sem{C}$.

\paragraph{$\bigoplus\beta$}
The $\beta$ rule for $\bigoplus$ is given by,
\[
  \semg{\letin {\sigma~x~a} {\sigma~M~a'} {e'}} = \semg{e'[M/x,a'/a]}
\]
\begin{proof}
  \begin{align*}
  \llbracket \mathsf{let}\, {\sigma~x~a} = & {\sigma~M~a'} \, \mathsf{in}\, {e'} \rrbracket \gamma \\
  & = {[ \sem{e'}(\gamma , x) ]_{(x : \semg{X})}} \circ \coprod_{x:\semg{X}}(m^{-1}) \circ d \circ (id
  \otimes \semg{\sigma~M~a'}) \otimes id \circ m \otimes id \circ m \\
  & = {[ \sem{e'}(\gamma , x) ]_{(x : \semg{X})}} \circ \coprod_{x:\semg{X}}m^{-1} \circ d \circ (id
  \otimes (i_M \circ \semg{a'})) \otimes id \circ m \otimes id \circ m \\
  & = {[ \sem{e'}(\gamma , x) ]_{(x : \semg{X})}} \circ \coprod_{x:\semg{X}} m^{-1} \circ d \circ (id \otimes i_M) \otimes id \circ (id \otimes  \lambda) \otimes id \circ m \otimes id \circ m \\
  & = {[ \sem{e'}(\gamma , x) ]_{(x : \semg{X})}} \circ \coprod_{x:\semg{X}} m^{-1} \circ i_M \circ (id \otimes  \lambda) \otimes id \circ m \otimes id \circ m \\
  & = \sem{e'}(\gamma , M) \circ m^{-1} (id \otimes  \lambda) \otimes id \circ m \otimes id \circ m \\
  & = \sem{e'}(\gamma , M) \tag{coherence}
  \end{align*}

  \begin{align*}
    \semg{e'[M/x,a'/a]}
    &= \semg{e'} \circ \semg{M/x,a'/a} \tag{\cref{lem:subst}} \\
    &= \sem{e'}(\gamma , x)
  \end{align*}
\end{proof}

\paragraph{$\bigoplus\eta$}
It suffices to show
\[
  \semg{\letin {\sigma~x~a} {c'} {f[(\sigma~x~a)/c]}} = \semg{f[c'/c]} = \semg{f}
  \]
First, expanding the left hand side, we have.
\begin{proof}
  \begin{align*}
    \llbracket \mathsf{let}\, {\sigma~x~a} = {c'} & \, \mathsf{in}\, {f[(\sigma~x~a)/c]} \rrbracket \gamma \\
    &= [ \semg{f[(\sigma~x~a)/c]} ]_{(x : \semg{X})} \circ \coprod_{x:\semg{X}}(m^{-1}) \circ d \circ (id \otimes \lambda) \otimes id \circ m \otimes id \circ m\\
    &= [ \semg{f} \circ (id \otimes i_x) \otimes id ]_{(x : \semg{X})} \circ \coprod_{x:\semg{X}}(m^{-1}) \circ d \circ (id \otimes \lambda) \otimes id \circ m \otimes id \circ m\\
    &= \semg{f} \circ [ (id \otimes i_x) \otimes id \circ (m^{-1})]_{(x : \semg{X})} \circ d \circ (id \otimes \lambda) \otimes id \circ m \otimes id \circ m\\
  \end{align*}
Since the domain has the universal property of a coproduct (due to distributivity), to prove
this is equal to $\semg{f}$, it is sufficient to prove they are equal
when composed with the injections:

\begin{align*}
  \semg{f} &\circ [ (id \otimes i_x) \otimes id \circ (m^{-1})]_{(x : \semg{X})} \circ d \circ (id \otimes \lambda) \otimes id \circ m \otimes id \circ m \circ (id \otimes i_y) \otimes id\\
  &= \semg{f} \circ [ (id \otimes i_x) \otimes id \circ (m^{-1})]_{(x : \semg{X})} \circ d  \circ (id \otimes i_y) \otimes id \circ (id \otimes \lambda) \otimes id \circ m \otimes id \circ m \tag{naturality}\\
  &= \semg{f} \circ [ (id \otimes i_x) \otimes id \circ (m^{-1})]_{(x : \semg{X})} \circ i_y \otimes id \circ (id \otimes \lambda) \otimes id \circ m \otimes id \circ m \tag{naturality}\\
  &= \semg{f} \circ (id \otimes i_y) \otimes id \circ (m^{-1}) \otimes id \circ (id \otimes \lambda) \otimes id \circ m \otimes id \circ m \\
  &= \semg{f} \circ (m^{-1}) \otimes id \circ (id \otimes \lambda) \otimes id \circ m \otimes id \circ m \circ (id \otimes i_y) \otimes id \\
  &= \semg{f} \circ (id \otimes i_y) \otimes id \tag{coherence}\\
\end{align*}
\end{proof}

\subsubsection{Equalizer}
\paragraph{Equalizer Introduction}
$\semg{\equalizerin{e}} : \semg{\Delta} \to \semg{\equalizer e f g}$ where
$\semg{e} : \semg{\Delta} \to \semg{A}$ and
$\semg{f} \circ \semg{e} = \semg{g} \circ \semg{e}$. By the universal property
of the equalizer the preceding equality induces a unique morphism
$\semg{\Delta} \to Eq(\semg{f} , \semg{g}) = \semg{\equalizer e f g}$. Define
$\semg{\equalizerin{e}}$ to be this map.

\paragraph{Equalizer Elimination}
$\semg{\equalizerpi e} : \semg{\Delta} \to \semg{A}$ is defined using the map
$\pi_{eq}$ from $Eq(\semg{f} , \semg{g})$ to the domain of $f$ and $g$.
\[
  \semg{\equalizerpi e} = \pi_{eq} \circ \semg{e}
\]
\paragraph{Equalizer $\beta$}
The $\beta$ rule for $\equalizer e f g$ is given as,
\[
\semg{\equalizerpi {\equalizerin e}} = \semg{e}
\]
where $\semg{f} \circ \semg{e} = \semg{g} \circ \semg{e}$.
\begin{proof}
  In $\semcat$ the universal property of $Eq(\semg{f}, \semg{g})$ implies that
  the following diagram commutes, implying the $\beta$ rule.

\begin{center}
\begin{tikzcd}
	{Eq(\semg{f},\semg{g})} & {\semg{A}} & {\semg{B}} \\
	{\semg{\Delta}}
	\arrow["{\pi_{eq}}", from=1-1, to=1-2]
	\arrow["{\semg{f}}", shift left, from=1-2, to=1-3]
	\arrow["{\semg{g}}", shift right, swap, from=1-2, to=1-3]
	\arrow["{\semg{\equalizerin e}}", from=2-1, to=1-1]
	\arrow["{\semg{e}}"', from=2-1, to=1-2]
\end{tikzcd}
\end{center}
\end{proof}

\paragraph{Equalizer $\eta$}

The $\eta$ rule for $\equalizer e f g$ is given as,
\[
\semg{\equalizerin {\equalizerpi e}} = \semg{e}
\]
\begin{proof}
  Likewise, the universal property of $Eq(\semg{f}, \semg{g})$ implies $\eta$
  rule via this diagram.

\begin{center}
\begin{tikzcd}
	{Eq(\semg{f},\semg{g})} & {\semg{A}} & {\semg{B}} \\
	{\semg{\Delta}}
	\arrow["{\pi_{eq}}", from=1-1, to=1-2]
	\arrow["{\semg{f}}", shift left, from=1-2, to=1-3]
	\arrow["{\semg{g}}", shift right, swap, from=1-2, to=1-3]
	\arrow["{\semg{e}}", from=2-1, to=1-1]
	\arrow["{\semg{\equalizerpi e}}"', from=2-1, to=1-2]
\end{tikzcd}
\end{center}
\end{proof}

\subsection{Grammar Semantics Respects Additional Axioms}
We verify that the denotational semantics validates each of the axioms we have
assumed.

\paragraph{\cref{ax:dist}} Distributivity
\begin{theorem}[\Agda]
  In $\Grammar$, $\semg{\lamblto{e}{\letin{\sigma\,f\,e'}{e}{\withlamb{x}{e'\,.\pi\,x}}}}$ has
  an inverse.
\end{theorem}
\begin{proof}
  Distributivity is true in the denotational semantics, as the category $\Grammar$ is a topos, which are well-known to be distributive.
  The following map forms the desired inverse,
  \[
    \lamb {\gamma} {\lamb {w} {\lamb {p} {(\lamb {\left( \lamb {x} {(p~x).fst} \right)} , \lamb {x} {(p~x).snd})}}}
  \]
\end{proof}

\paragraph{\cref{ax:disjointness}} $\sigma$-Disjointness
\begin{theorem}[\Agda]
  In $\Grammar$, $\semg{\sigma x}$ and $\semg{\sigma x'}$ are disjoint for
  $x \neq x'$.
\end{theorem}
\begin{proof}
  This is trivially true, as the denotation of linear sum types is as $\Sigma$
  types in $\Grammar$. For any input, the first projections of $\semg{\sigma x}$
  is $x$ and the first projection of
  $\semg{\sigma x'}$ is $x'$. Because $x \neq x'$, the images of
  $\semg{\sigma x}$ and $\semg{\sigma x'}$ cannot agree.
\end{proof}

\paragraph{\cref{ax:string-top}} String is strongly equivalent to $\top$.
\begin{theorem}[\Agda]
  In $\Grammar$, $\semg{!}$ has an inverse.
\end{theorem}
\begin{proof}
  Because $\semg{\top}w$ is a singleton set for all $\gamma$ and $w$, it suffices to show that
  $\semg{\StringGram}w$ is likewise a singleton.

  First, we prove by induction on $w$ that $\semg{\StringGram}w$ is a retract of
  $\semg{\LinSigTy{w}{\StringSem}{\internalize{w}}}w$, a sum over a nonlinear
  type of strings. Then, again by
  induction on $w$, we show that
  $\semg{\LinSigTy{w}{\StringSem}{\internalize{w}}}w$ is isomorphic to
  $\semg{\top}w$.

  Each $\semg{\internalize{w}}w'$ is inhabited if and only if $w$ is equal to
  $w'$, and the parses of $w$ for $\semg{\internalize{w}}$ are unique. That is,
  $\semg{\LinSigTy{w}{\StringSem}{\internalize{w}}}w$ is a singleton set.
  So, $\semg{\StringGram}w$ is a retract of a singleton set, and is itself
  singleton. Thus, $\semg{\StringGram}w \cong \semg{\top}w$ in $\Set$.
\end{proof}

\fi
\end{document}